 \documentclass[10pt, twocolumn]{IEEEtran}
\usepackage{mathtools}

\newenvironment{floateq}{
\begin{figure*}[!t]
\normalsize

}{

\hrulefill
\vspace*{4pt}
\end{figure*}
}

 \usepackage{enumitem}
\usepackage{color}
\usepackage{polpack}
\usepackage{epstopdf}
\usepackage{xfrac}
\usepackage{bm}
\usepackage{accents}

\newcommand{\barpi}{\mathbf{\widehat{\text{$\Pi$}}}}

\newcommand{\ordpiv}{\mathbf{\accentset{\circ}{\text{$\Pi$}}}}
\newcommand{\eqnsize}{\normalsize}

\IEEEoverridecommandlockouts

\begin{document}

\allowdisplaybreaks[4]
\title{Multi-Access Communications with Energy Harvesting: A Multi-Armed Bandit 
Model and the Optimality of the Myopic Policy}

\author{\IEEEauthorblockN{Pol Blasco and Deniz G{\"u}nd{\"u}z} \\
\IEEEauthorblockA{Imperial College London,  UK\\
Emails: \{p.blasco-moreno12, d.gunduz\}@imperial.ac.uk}}

\maketitle
\begin{abstract}
A multi-access wireless network with $N$ transmitting nodes, each 
equipped with an energy harvesting (EH) device and a rechargeable battery of 
finite capacity, is studied. At each time slot (TS) a node is 
\emph{operative} with a certain probability, which may depend on the 
availability of data, or the state of its channel. The energy 
arrival process at each node is modelled as an independent two-state Markov 
process, such that, at each TS, a node either harvests one unit of energy, or none. At each TS a subset of the nodes is scheduled by the access point (AP). The 
scheduling policy that maximises the total throughput is studied assuming that the AP does not know the states of either the EH processes or the batteries. The 
problem is identified as a restless multi-armed bandit (RMAB) problem, and an 
upper bound on the optimal scheduling policy is found. Under certain assumptions 
regarding the EH processes and the battery sizes, the optimality of the myopic policy (MP) is proven. For the general case, the performance of MP is compared numerically to the 
upper bound. 
\end{abstract}

\section{Introduction} 
Low-power wireless networks, such as machine-to-machine and wireless sensor 
networks, can be complemented  with energy harvesting (EH) technology to extend 
the network lifetime. A  low-power wireless node has a limited lifetime constrained by the battery size; but when complemented with an EH device and a rechargeable battery, its lifetime can be prolonged significantly. However, energy availability at the EH nodes is scarce, and, due 
to the random nature of the energy sources, energy arrives at random times and in 
arbitrary amounts. Hence, in order to take the most out of the scarce energy, it is important to optimise the scheduling policy of the wireless network.

Previous research on EH wireless networks can be grouped into three, based on 
the information available regarding the random processes governing the system 
\cite{EH:Gunduz2014}. In the offline optimization framework, availability of 
non-causal information on the exact realizations of the random processes 
governing the system is assumed at the transmitter \cite{Yang:TC:10}, 
\cite{Devillers:JCN:11}. In the online optimization framework 
\cite{hj:Wang2012,EH:Aprem2013,hj:Lei2009,EH:Micheusi2013, 
Iannello2012a,P.Blasco2013,ehc:O.M.Gul2014,Blasco2013}, the statistics governing 
the random processes are assumed to be available at the transmitter, and their 
realizations are known only causally. The EH communication system is modeled as 
a Markov decision process (MDP) \cite{hj:Wang2012}, or as a partially observable 
MDP (POMDP) \cite{EH:Aprem2013}, and dynamic programming 
(DP)~\cite{Bertsekas2005} can be used to optimise the EH communication system 
numerically. In many practical applications, the state space of the 
corresponding MDPs and POMDPs is large, and DP becomes computationally 
prohibitive~\cite{Littman1995a}, and the numerical results of DP do not provide 
much intuition about the structure of the optimal scheduling policy. In order to 
avoid complex numerical optimisations it is important to characterize the 
behaviour of the optimal scheduling policy and identify properties about its 
structure; however, this is possible only in some special cases 
\cite{Iannello2012a,P.Blasco2013,hj:Lei2009}. In the learning optimization 
framework, the knowledge about the system behaviour is further relaxed, and even 
the statistical knowledge about the random processes governing the system is not 
assumed, and the optimal policy scheduling is learnt over 
time~\cite{Blasco2013}.

We study online scheduling of low-power wireless 
nodes by an access point (AP). The nodes are equipped with EH devices, and powered by 
rechargeable batteries. At each time slot (TS) a node is operative with a 
certain probability, which may depend on the channel conditions or the 
availability of data at the node. The EH process at each node is modelled as an 
independent Markov process, and at each TS, a node either harvests one unit of 
energy or does not harvest any. The AP is in charge of scheduling, at each TS,  
the EH nodes to the available orthogonal channels. A node transmits only when it is scheduled and is operative at the same time. Hence, at each TS the AP learns the EH process states and 
battery levels of the operative nodes that are scheduled, but does not receive 
any information about the other nodes.  The AP is interested in maximising the 
expected sum throughput within a given time horizon. This problem can be model 
as a POMDP and solved numerically using DP at the expense of a 
high-computational cost. Instead, we model it as a restless multi-armed bandit 
(RMAB) problem  \cite{RMAB:Whittle1988}, and prove the optimality of a 
low-complexity policy in two special cases. Moreover, by relaxing the constraint 
on the number of nodes that the AP can schedule at each TS, we obtain an upper 
bound on the performance of the optimal scheduling policy. Finally, the 
performance of the low complexity policy is compared to that of the upper bound 
numerically.  The main technical contributions of the paper are summarised as 
follows:
\begin{itemize}
 \item We show the optimality of a MP if the nodes do not harvest 
energy and transmit data at the same time, and the EH process is affected by the 
scheduling policy. 
 \item We show the optimality of MP if the nodes do not have 
batteries and can transmit only if they have harvested energy in the previous 
TS. 
 \item We provide an upper bound on the performance for 
the general case by relaxing the constraint on the number of nodes that can be 
scheduled at each~TS.
 \item We show numerically that MP performs close to the upper 
bound for the general case. 
\end{itemize}

The rest of this paper is organized as follows. Section \ref{sec:rel_work} is 
dedicated to a summary of the related literature. In 
Section~\ref{sec:system_model}, we present the  EH wireless multi-access network 
model. In Sections~\ref{sec:case1} and~\ref{sec:case2} we characterize 
explicitly the structure of the optimal policy that maximises the sum throughput 
for two special cases. In Section~\ref{sec:upp}, we provide an upper bound on 
the performance.  Finally, in Section~\ref{sec:numerical} we compare the 
performance of MP with that of the upperbound through numerical 
analysis. Section \ref{sec:con} concludes the paper.

\section{Related Work}\label{sec:rel_work}

There is a growing research interest in EH wireless communication systems, and 
in particular, in developing scheduling policies that exploit the scarce 
harvested energy in the most efficient manner. In large EH wireless networks, 
since numerical optimization is computationally prohibitive, it is important to 
characterise the optimal scheduling policy explicitly, or certain properties of it.

In \cite{hj:Lei2009}, the authors assume that the data packets arrive at the EH 
transmitter as a Poisson process, and each packet has an intrinsic value 
assigned to it, which also is a random variable. The optimal transmission policy 
that maximizes the average value of the received packets at the destination is 
proven to be a threshold policy. However, the values of the thresholds 
have to be computed using numerical techniques, such as DP or linear programming 
(LP). Reference~\cite{EH:Micheusi2013} extends the problem in~\cite{hj:Lei2009} 
to the multi-access scenario.

Multi-access in EH wireless networks with a central scheduler, static channels 
and backlogged nodes has been studied in  \cite{Iannello2012a, P.Blasco2013, 
ehc:O.M.Gul2014}. The central scheduler in \cite{Iannello2012a} does not 
know the battery levels or the states of the EH processes at the nodes. Assuming 
that the nodes have unit size batteries, the system is modeled as an RMAB, and 
MP, which has a round robin (RR) structure, is shown to maximise 
the sum throughput. Reference \cite{P.Blasco2013} considers nodes with batteries 
of arbitrary capacity, and  MP is found to be optimal in two special cases. In 
contrast to the present paper, \cite{P.Blasco2013} considers static channels and 
backlogged nodes, and the optimality proof exploits the RR structure of MP. In  
\cite{ehc:O.M.Gul2014}, considering infinite-capacity batteries, an asymptotically 
optimal policy is proposed. 

The problem studied in this paper is modeled as an RMAB problem. In the classic 
RMAB  problem there are several arms, each of which is modelled as a Markov 
chain \cite{RMAB:Whittle1988}. The states of the arms are unknown, and at each TS an 
arm is played. The played arm reveals its state and yields a reward, which is 
a function of the state. The objective is to find a policy that maximises the 
total reward over time. RMAB problems have been shown to be, in general, PSPACE 
hard \cite{RMAB:Papadimitriou1994}, and our knowledge on the structure of the 
optimal policy for general RMAB problem is limited. 

Recently, the RMAB model has been used to study channel access and cognitive 
radio problems, and new results on the optimality of MP have been obtained 
\cite{RMAB:Ahmad2009,RMAB:Ahmad2009a,RMAB:Mansourifard2012,RMAB:Wang2012,
RMAB:Liu2010}. The structure  and the optimality of MP  is proven in 
\cite{RMAB:Ahmad2009} and \cite{RMAB:Ahmad2009a} for single and multiple plays, 
respectively, under certain conditions on the Markov transition probabilities. 
In \cite{RMAB:Mansourifard2012} the optimality of MP is shown for a general 
class of monotone affine reward functions, which include arms with arbitrary 
number of states.  The optimality of MP is proven in \cite{RMAB:Wang2012}  when 
the arms' states follow non-identical Markov chains. The case of imperfect 
channel detection is studied in \cite{RMAB:Liu2010},  and MP is found to be 
optimal when the false alarm probability of the channel state detector is below 
a certain value. 

\section{System Model}\label{sec:system_model}

We consider an EH wireless network with $N$ EH nodes and one AP, as depicted in 
Figure~\ref{fig:system_model}. Time is divided into TSs of constant duration, 
and  the AP is in charge of scheduling $K$ of the $N$ nodes to the $K$ available 
orthogonal channels at each TS. A node is \emph{operative}  at each TS with a 
fixed probability $p$ independent over TSs and nodes, and 
\emph{inoperative}  otherwise. We consider that a node is in the operative state 
if it has a data packet to transmit in its buffer and its channel to the AP is 
in a good state, while it is inoperative otherwise even if it is scheduled to a 
channel. The EH process is modelled as a Markov chain, which can be either in 
the harvesting or in the non-harvesting state, denoted by states $1$ and $0$, 
respectively. We denote by $p_{ij}$ the transition probability from state $i$ to 
$j$, and assume that $p_{11}\geq p_{01}$, that is, the EH process is positively 
correlated in time, and hence, if the EH process is in state $i$, it is more 
likely to remain in state $i$ than switching to the other state. We denote by 
$E^s_{i}(n)$ and $E^h_{i}(n)$ the state of the EH process and the amount of 
energy harvested by node $i$, respectively, in TS $n$. The energy harvested in 
TS $n$ is available for transmission in TS $n+1$. We assume that one fundamental 
unit of energy is harvested when the Markov process makes a transition to the 
harvesting state, that is, $E_i^h(n)=E_i^s(n+1)$\footnote{Our results can be generalised to  a broader class of two-state Markovian EH processes in which the amount 
of energy harvested in each state is an independent and 
identically distributed random variable, and the expected amount of 
harvested energy in the harvesting state is larger than that in the 
non-harvesting state. However, the studied EH model captures the 
random nature of the energy arrivals, and is also considered in 
\cite{hj:Wang2012,Iannello2012a,P.Blasco2013,Blasco2013}. }. Each node is 
equipped with a battery of capacity $B$, and we denote by $B_{i}(n) \in 
\{0,\ldots,B\}$ the amount of energy stored in the battery of node $i$ at the 
beginning of TS $n$.    The state of node $i$ in TS $n$, $S_i(n)$,  is given by 
its battery and EH process states, $S_i(n)=(E^s_i(n),B_i(n))\in\{0,1\}\times\{0,\ldots,B\}$.  The system state is 
characterized by the joint states of all the nodes.

\begin{figure}
  \centering
  \includegraphics[width=0.45\textwidth]{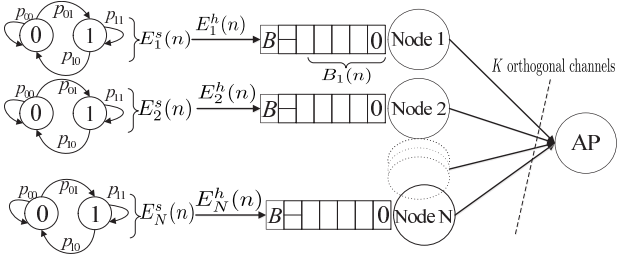}\\  
 \caption{System model with $N$ EH nodes with finite size batteries and $K$ orthogonal channels.}\label{fig:system_model}
\end{figure}

The system functions as follows: At the beginning of each TS, the AP schedules 
$K$ out of $N$ nodes, such that a single node is allocated to each orthogonal channel.  
When a node is scheduled, if it is operative in that TS, i.e., it has data to 
transmit and its channel is in a good state, it transmits a data packet as well 
as the current state of its EH process to the AP. If 
it is not operative it transmits a status beacon to the AP, and backs off. We 
say that a node is \emph{active} in a TS if it is scheduled by the AP and is 
operative; and hence, it transmits a data packet to the AP, otherwise we say 
that the node is \emph{idle} in this TS, that is, the node is not scheduled or 
it is scheduled, but it is not operative.  We denote by $\mathcal{K}(n)$ and  
$\mathcal{K}^a(n)$ the set of nodes scheduled by the AP, and the set of active 
nodes in TS $n$, respectively, where $\mathcal{K}^a(n)\subseteq 
\mathcal{K}(n)$. 

We assume that the transmission rate is a linear function of the transmit power, which is an accurate approximation in the low power regime. When the power-rate function is linear, the total 
number of bits transmitted to the AP is maximised when 
an active node transmits at a constant power throughout the TS, using all its 
energy. To simplify 
the notation we normalise the power-rate function such that the number of bits 
transmitted within a TS is equal to the energy used for 
transmission.  Then the expected throughput in TS $n$ is
\eqnsize\begin{equation}  \label{eq:throughput}
 R(\mathcal{K}(n))=\expected{~~\sum_{\mathclap{i\in 
\mathcal{K}^a(n)}}B_i(n)}{}=p  \sum_{\mathclap{i\in \mathcal{K}(n)}}B_i(n).
\end{equation}\normalsize
The objective of the AP is to schedule the best set of nodes, $\mathcal{K}(n)$, 
at each TS in order to maximize (\ref{eq:throughput}), without knowing which 
nodes are operative, the battery levels, or the EH states. The only information 
the AP receives is the EH state of the active nodes at each TS. Note 
that the AP also knows the battery state of the active nodes after transmission 
since they use all their energy. 

A scheduling policy is an algorithm that schedules nodes at each TS $n$, based 
on the previous observations of the EH states and battery levels. The objective 
of the AP is to find the scheduling policy $\mathcal{K}(n)$, $\forall n \in 
[1,T]$, that maximizes the total discounted throughput, given by
\eqnsize
\begin{equation}   \label{eq:opt_problemv1}
\begin{aligned}
\max_{\{\mathcal{K}(n)\}_{n=1}^{T}} &~~ \sum_{n=1}^T \beta^{n-1} 
R(\mathcal{K}(n)),\\
\text{s.t. } &  B_{i}(n+1) =\min \{B_{i}(n) \\
&+ E^h_{i}(n),B \} \cdot \mathds{1}_{i \notin \mathcal{K}^a(n)}
+ E^h_{i}(n) \cdot \mathds{1}_{i \in \mathcal{K}^a(n)},
\end{aligned}
\end{equation}
\normalsize
where $0 <\beta  \leq1$ is the discount factor, and $\mathds{1}_{a}$ is  the 
indicator function, defined as  $\mathds{1}_{a}=1$ if $a$ is true, and 
$\mathds{1}_{a}=0$, otherwise. 

If the AP is informed on the current state of all the nodes at each TS, 
the problem would be formulated as an MDP, and  solved  using LP or DP 
\cite{Bertsekas2005}. However, in practice transmitting all the nodes' states to 
the AP introduces further overhead and energy consumption; and hence, is not 
considered here. Accordingly, the appropriate model for our problem is a 
POMDP. It can be shown that a sufficient statistic for optimal decision making 
in a POMDP is given by the conditional probability of the system states given 
all the past actions and observations, which, in our problem, depends only on 
the number of TSs each node has been idle for, and on the realisation of each 
node's EH state last time it was active. Hence, we can reformulate the POMDP 
into an equivalent MDP with an extended state space. The belief states, that is, 
the states in the equivalent MDP, are characterized by all the past actions and 
observations. We denote by $l_i$ and $h_i$ the number of TSs that node $i$ has 
been idle for, and the state of the EH process the last time it was active, 
respectively. The belief state of node $i$, $s_i(n)$, is given by 
$s_i(n)=(l_i,h_i)$, and the belief state of the whole system is the joint belief 
states of all the nodes. In TS $n$, the belief state of node~$i$ is updated as 
$s_i(n+1)=(0,E^s_i(n))$,  if $i \in \mathcal{K}^a(n)$, and  as 
$s_i(n+1)=(l_i+1,h_i)$, otherwise.   That is, at each TS,  $l_i$ is set to $0$ 
if node~$i$ is active, and increased by one if it is idle. In principle, since 
the number of TSs a node can be idle is unbounded, the state space of the 
equivalent MDP is infinite, and hence, the POMDP in (\ref{eq:opt_problemv1}) is 
hard to solve numerically. In Sections~\ref{sec:case1}~and~\ref{sec:case2}, we focus on 
two particular settings, and show the existence of optimal low-complexity scheduling policies
under certain assumptions.

\section{Non Simultaneous Energy Harvesting and Data 
Transmission}\label{sec:case1}

In this section we assume that the nodes are not able to harvest energy and 
transmit data simultaneously,  and that if node~$i$ is active in TS~$n-1$, then 
its EH state in TS~$n$, $E_i^s(n)$, is either $0$ or $1$ with probabilities 
$e_0$ and $e_1$, respectively, independent of the EH state in TS~$n-1$, where 
$e_0 \leq \frac{p_{10}}{p_{01}+p_{10}}$. These assumptions may account for nodes 
equipped with electromagnetic energy harvesters in which the same antenna is 
used for harvesting as well as transmission; and hence, it is not possible to 
transmit data and harvest energy simultaneously, and the RF hardware has to be 
reset into the harvesting mode after each transmission.

Since the EH process is reset when a node 
transmits, the EH process states of active nodes are not relevant. As a 
consequence, the belief state of a node, $s_i(n)$, is characterized only by the 
number of TSs the node has been idle for, $l_i$. There is a one-to-one 
correspondence between $l_i$ and the expected battery level of 
node $i$; therefore, we redefine the belief state, $s_i(n)$, as the expected 
battery level of node $i$ in TS $n$, normalised by the battery capacity. The 
expected throughput in~(\ref{eq:throughput}) can be rewritten as 
\eqnsize\begin{equation}  \label{eq:reward}
 R(\mathcal{K}(n))=p B \sum_{\mathclap{i\in \mathcal{K}(n)}}s_i(n).
\end{equation}\normalsize
Notice that $s_i(n)$ in \eqref{eq:reward} is normalised, i.e., $s_i(n)\in[0,1]$. 


Due to the Markovity of the EH processes, the future belief state is only 
a function of the current belief state and the scheduling policy.  If a node is 
active in TS $n$, since it uses all its energy and does 
not harvest any, the belief state is set to $0$ in TS $n+1$. If a node is not 
active in TS~$n$, then the belief state evolves according to the belief state 
transition function $\tau(\cdot)$. The belief state of node $i$ in TS $n+1$ is 
 \eqnsize\begin{equation} \label{eq:believe_state} 
 s_i(n+1)=\left \{ \begin{array}{lll}
\tau(s_i(n)) & \text{ if~}  i \notin \mathcal{K}^a(n), \\
0 	     & \text{ if~}  i \in \mathcal{K}^a(n).  \\
\end{array} \right .
\end{equation} \normalsize

\begin{property}\label{prop1}
The belief state transition function, $\tau(\cdot)$, is a monotonically 
increasing contracting map, that is,  $\tau(s_i(n))>\tau(s_j(n))$ if  
$s_i(n)>s_j(n)$,  and $ \|\tau(s_i(n)) -\tau(s_j(n)) \| \leq \|s_i(n)-s_j(n)\|$.
\end{property}
\begin{proof}
The proof is given in Appendix~\ref{app:contracting_map}.
 \end{proof}

Note that the assumption $p_{11} \geq p_{01}$ is a necessary condition for 
Property~1.
We denote by $\mathbf{s}(n)=(s_1(n), \ldots, s_N(n))$ the belief vector in TS 
$n$, which contains the belief states of all the nodes, and by  
$\mathbf{s}_\mathcal{E}(n)$ the belief vector of the nodes in set $\mathcal{E}$. 
For the sake of clarity we drop the $n$ from $\mathbf{s}(n)$ and  
$\mathbf{s}_\mathcal{E}(n)$  when the time index is clear from the context.  We 
denote the expected throughput by $R(\mathbf{s}_\mathcal{E})$ if the belief 
vector is $\mathbf{s}$ and nodes in $\mathcal{E}$ are scheduled.
 
The probability that a particular set of nodes, $\mathcal{K}^a(n) \subseteq 
\mathcal{K}(n)$, is active while the rest of the scheduled nodes remain idle in 
TS $n$ is a function of the cardinality of $\mathcal{K}^a(n)$ and the 
probability that a node is operative, $p$. For $a\triangleq|\mathcal{K}^a(n)|$ 
we denote this probability by 
\eqnsize\begin{equation}  \label{eq:kchannels}
q(a,K)\triangleq (1-p)^{K-a}p^{a}.
\end{equation} \normalsize

The AP is interested in finding the scheduling policy $\pi$, which
schedules the nodes according to $\mathbf{s}(n)$, that is 
$\mathcal{K}(n)=\pi(\mathbf{s}(n))$, such that the expected throughput over 
the time horizon $T$ is maximised. The associated optimization problem is 
expressed through the Bellman value functions,
\eqnsize\begin{equation}  \label{eq:val}
\begin{array}{rcl}
V_n^\pi(\mathbf{s})&=&  R(\mathbf{s}_{\pi(\mathbf{s})}) +\beta 
\displaystyle\sum_{\mathclap{\mathcal{E}\subseteq \pi(\mathbf{s})}} 
q(|\mathcal{E}|,K) \\
&&  \times V_{n+1}^\pi( (s_1(n+1),\ldots,s_j(n+1)=0, \\
&& \ldots,s_i(n+1)=\tau(s_i(n)), \ldots)), 
\end{array}
\end{equation} \normalsize
where the sum is over all possible sets of active nodes, $\mathcal{E}$, among 
the scheduled nodes, $\mathcal{K}(n)=\pi(\mathbf{s}(n))$, and nodes $j$ and $i$ 
are active and idle, respectively.  The optimal policy, $\pi^*$, is the one that 
maximises (\ref{eq:val}).

\subsection{Definitions}

\begin{defn} 
At TS $n$ the \textbf{myopic policy (MP)} schedules the $K$ nodes that maximise the expected instantaneous reward function, $R(\cdot)$. For the reward function in 
(\ref{eq:reward}) the MP schedules the $K$ nodes with the highest belief states. 
 \end{defn}

MP schedules the nodes similarly to a round robin (RR) policy that orders the 
nodes according to the time they have been idle for, and at each TS schedules 
the nodes with the highest idle time values. If a node is active in this TS, it 
is sent to the bottom of this ordered list in the next TS. If a node  is idle it 
moves forward in the order. Notice that  due to the monotonicity of 
$\tau(\cdot)$ the order of the idle nodes is preserved.

We denote by $\mathbf{s}_\Pi=(s_{\Pi(1)},\ldots,s_{\Pi(N)})$, the permutation of 
the vector $\mathbf{s}$, where $\Pi(\cdot)$ is a permutation function, by  
$\mathbf{s}_\Pi^{K}=(s_{\Pi(1)},\ldots,s_{\Pi(K)})$ the vector containing the 
first $K$ elements of $\mathbf{s_\Pi}$, and by $\mathcal{S}_\Pi^{K}=\{\Pi(1), 
\cdots, \Pi(K) \}$ the set of indices of the nodes in positions from $1$ to $K$ 
in vector  $\mathbf{s}_\Pi$. We say that a vector is ordered if its elements are 
in decreasing order. We denote by $\ordpiv$ the permutation that orders a 
vector, that is,  the vector $\mathbf{s}_\ordpiv$ is ordered, i.e., 
$s_{\ordpiv(1)}\geq s_{\ordpiv(2)}\geq \ldots \geq s_{\ordpiv(N)}$. We denote 
the vector operator that first orders the vector $\mathbf{s}_\mathcal{E}$ of 
$|\mathcal{E}|$ components, and then applies $\tau(\cdot)$ to each of the 
components of the resulting vector by 
$\mathrm{T}(\mathbf{s}_\mathcal{E})\triangleq (\tau(s_{\ordpiv(1)}),\cdots,\tau(s_{
\ordpiv(|\mathcal{E}|)}))$, with $\ordpiv(i) \in \mathcal{E}, 1 \leq i \leq 
|\mathcal{E}|$.  Note that due to the monotonicity of $\tau(\cdot)$ the vector 
$\mathrm{T}(\mathbf{s}_\mathcal{E})$ is always ordered. Finally, we denote the 
zero vector of length $k$ by $\mathbf{0}(k)$.

\begin{defn} \textit{Pseudo value function}, $W_n(\mathbf{s}_\Pi)$, is defined as
 \eqnsize\begin{equation}  \label{eq:pseudo_val}
\begin{array}{l}
W_n(\mathbf{s}_\Pi)\triangleq  R(\mathbf{s}_\Pi^{K}) +\beta \displaystyle 
\sum_{\mathclap{\mathcal{E}\subseteq \mathcal{S}_\Pi^{K}}}  q(|\mathcal{E}|,K) 
W_{n+1}(\left [\mathrm{T}(\mathbf{s}_{\overline{\mathcal{E}}})~ , 
\mathbf{0}(|\mathcal{E}|) \right ]), \\
W_T(\mathbf{s}_\Pi)\triangleq   R(\mathbf{s}_\Pi^{K}),
\end{array} 
\end{equation}\normalsize
where $[\cdot, \cdot]$ is the vector concatenation operator. 

$W_n(\cdot)$ is  characterized solely by the belief vector 
$\mathbf{s}$ and its initial permutation $\Pi$. In TS $n$, the first $K$ nodes 
according to permutation $\Pi$ are scheduled, and the nodes are scheduled 
according to MP thereafter.  The belief vector in TS $n+1$ is  
$\mathbf{s}_\ordpiv(n+1)= \left 
[\mathrm{T}(\mathbf{s}_{\overline{\mathcal{E}}}), 
\mathbf{0}(|\mathcal{E}|)\right]$, where $\mathcal{E}$ is the set of active 
nodes in TS $n$, and, since $\mathrm{T}(\cdot)$ implicitly orders the output 
vector, $\mathbf{s}_\ordpiv(n+1)$ is ordered. Hence, the nodes that are active 
in TS $n$  have belief state $0$  in TS $n+1$, and are moved to the rightmost 
position in the belief vector.  If vector $\mathbf{s}_\Pi$ is ordered, 
(\ref{eq:pseudo_val}) corresponds to the value function of MP, that is, 
corresponds to (\ref{eq:val}) where $\pi$ is MP. 
\end{defn}

\begin{defn}
A permutation $\Pi$ is an \textit{$i,j$-swap} of permutation $\barpi$ if 
$\Pi(k)=\barpi(k)$, for $\forall{k}\neq \{i,j\}$, and $\Pi(j)=\barpi(i)$ and 
$\Pi(i)=\barpi(j)$. That is, all the nodes but those in positions $i$ and $j$ 
are in the same positions in $\mathbf{s}_\Pi$ and $\mathbf{s}_\barpi$, and the 
nodes in positions $i$ and $j$ are swapped. 

A permutation $\Pi$ is an \textit{$i,j$-swap} if $\Pi(k)=k$, for $\forall k\neq 
\{i,j\}$, and $\Pi(i)=j$ and $\Pi(j)=i$. That is, all the nodes but those in 
positions $i$ and $j$ are in the same position in $\mathbf{s}$ and 
$\mathbf{s}_\Pi$, and the nodes in positions $i$ and $j$ are swapped.

\end{defn}

\begin{defn} A function $f(\mathbf{x})$, 
$f:\mathds{R}^k\rightarrow \mathds{R}$ and $\mathbf{x}=(x_1,\ldots,x_k)$, is 
said to be \emph{regular} if it is symmetric, monotonically increasing, and 
decomposable \cite{RMAB:Wang2012}.
\begin{itemize}
  \item $f(\mathbf{x})$ is \textit{symmetric} if   $f(\ldots, x_i, \ldots, 
x_j,\ldots)=f(\ldots, x_j, \ldots, x_i,\ldots)$.
  \item $f(\mathbf{x})$ is \textit{monotonically increasing} in each of its 
components, that is, if $x_j \leq {\tilde{x}_j}$ then   $f(\ldots, x_j, 
\ldots)\leq f(\ldots, {\tilde{x}_j}, \ldots)$.
  \item $f(\mathbf{x})$ is \textit{decomposable} if  $f(\ldots, x_j, 
\ldots) = x_j f(\ldots, 1, \ldots)+(1- x_j)f(\ldots, 0, \ldots)$.
   \end{itemize}
\end{defn}
\begin{defn}(Boundedness) A function $f(\mathbf{x})$, $f:\mathds{R}^k\rightarrow 
\mathds{R}$ and $\mathbf{x}=(x_1,\ldots,x_k)$, is said to be \emph{bounded} if 
$\Delta_{l}\! \leq \!f(\ldots, 1, \ldots)\!-\!f(\ldots, 0, \ldots)\!\leq\! 
\Delta_{u}$.
\end{defn}

We note that the expected throughput $R(\cdot)$ is a linear function of the 
belief vector, which has  bounded elements, and all the nodes that are 
scheduled have the same coefficient; hence, $R(\cdot)$ is a bounded regular 
function. The pseudo value function, $W_n(\cdot)$, is symmetric, that is, 
\eqnsize\begin{equation}  \label{eq:Wsymmetric}
W_n(\mathbf{s}_\Pi)=W_n(\mathbf{s}_\barpi),
\end{equation} \normalsize
where $\Pi$ is a $i,j$-swap permutation of  $\barpi$, and $j,i \leq K $ or $j, i 
> K$. To see this we can use the symmetry of $R(\cdot)$, and the fact that 
$\mathrm{T}(\cdot)$ orders the belief vector in decreasing order. 

\subsection{Proof of the optimality of MP}

We prove the optimality of MP under the assumptions that 
$\tau(\cdot)$ is a monotonically increasing contracting map\footnote{Our results 
can also be applied to the case in which the state transition function is a 
monotonically increasing contracting map with parameter $\alpha$, that is, 
$\tau(s_i(n))>\tau(s_j(n))$ if  $s_i(n)>s_j(n)$,  and $ \|\tau(s_i(n)) 
-\tau(s_j(n)) \| \leq \alpha \|s_i(n)-s_j(n)\|$, if $0 \leq \alpha\cdot\beta 
\leq 1$. }, and $R(\cdot)$ is a bounded regular function. Hence, the 
results in this section  can be applied to a boarder class of EH processes and 
reward functions than those studied in this paper.  

The proof is structured as follows: Lemma 
\ref{lem:opt1} gives sufficient conditions for the optimality of  MP in TS $n$, 
given that  MP is optimal from TS $n+1$ onwards.  In Lemma~\ref{lem:upb} we show 
that the difference between the pseudo value functions of two different 
vectors is bounded. In particular, we bound the difference between the value 
functions of two belief vectors $\mathbf{s}_\ordpiv$ and 
$\tilde{\mathbf{s}}_{\ordpiv}$, which are both ordered, and differ only for 
the belief state of node $i$.  In Lemma~\ref{lem:swap} we show that, under 
certain conditions,  the sufficient conditions for the optimality of MP given in 
Lemma~\ref{lem:opt1} hold. 


\begin{lem}\label{lem:opt1} 
Assume that MP is optimal from  TS $n+1$ until TS $T$. A sufficient condition 
for the optimality of MP in TS $n$ is
\eqnsize\begin{equation} 
W_n(\mathbf{s})\geq W_n(\mathbf{s}_\Pi),  \label{eq:lem1}
\end{equation}  \normalsize
for any $\Pi$ that is an $i,j$-swap, with $ s_{j}\geq s_{i}$ and  $j\leq i$.
\end{lem}

\begin{proof}
To prove that a policy is optimal,  we need to show that it 
maximizes~(\ref{eq:val}). By assumption MP is optimal from TS $n+1$ onwards; and 
hence, it is only necessary to prove that scheduling any set of nodes and 
following MP thereafter is no better than following MP directly in TS $n$. The 
value function corresponding to the latter policy is 
$W_n([\mathbf{s}_\mathcal{O},\mathbf{s}_{\overline{\mathcal{O}}}])$, where 
$\mathbf{s}_\mathcal{O}$ contains the $K$ nodes with the highest belief states 
in $\mathbf{s}$, and $\mathbf{s}_{\overline{\mathcal{O}}}$ contains the rest of 
the nodes not necessarily ordered. The value function corresponding to the 
former policy is 
$W_n([\mathbf{s}_{\mathcal{U}},\mathbf{s}_{\overline{\mathcal{U}}}])$, where 
$\mathbf{s}_{\mathcal{U}}$ contains the $K$ nodes scheduled in TS $n$, and 
$\mathbf{s}_{\overline{\mathcal{U}}}$ is the set of the remaining nodes. There 
exist at least a pair of nodes $s_i$ and $s_j$ such that, $j\in 
{\overline{\mathcal{U}}}$ and $j \notin {\overline{\mathcal{O}}}$, $i \in 
{\mathcal{U}}$ and $i \notin {\mathcal{O}}$, and $s_j \geq s_i$. By swapping 
each pair of such nodes, that is, swapping $j\in {\overline{\mathcal{U}}}$ for 
$i \in {\mathcal{U}}$, we can obtain 
$W_n([\mathbf{s}_\mathcal{O},\mathbf{s}_{\overline{\mathcal{O}}}])$ from 
$W_n([\mathbf{s}_{\mathcal{U}},\mathbf{s}_{\overline{\mathcal{U}}}])$ through a 
cascade of inequalities  using (\ref{eq:lem1}). Accordingly, 
$W_n([\mathbf{s}_\mathcal{O},\mathbf{s}_{\overline{\mathcal{O}}}])$ is an upper 
bound for any 
$W_n([\mathbf{s}_{\mathcal{U}},\mathbf{s}_{\overline{\mathcal{U}}}])$, and, 
hence,~MP~is~optimal. 
\end{proof}

Lemma~\ref{lem:opt1} shows that, 
 under certain conditions, the optimality of MP can be established through the 
pseudo value function. In particular, under the conditions of 
Lemma~\ref{lem:opt1}, if swapping a node in the belief vector with another node 
with a lower position and a lower belief state does not decrease the pseudo 
value function, then MP is optimal.  
 
\begin{lem} \label{lem:upb} Consider a pair of belief vectors $\mathbf{s}$ and 
$\tilde{\mathbf{s}}$, which differ only in one element, that is, 
$s_i=\tilde{s}_i$ for $\forall i \neq j$ and $s_j \geq \tilde{s}_j$. If  
$R(\cdot)$ is a bounded regular function, $\tau(\cdot)$ a monotonically 
increasing contracting map, and $ \beta \leq 1$, then we have  
\eqnsize\begin{equation}  \label{eq:upb}
W_n(\mathbf{s}_\ordpiv)- W_n(\tilde{\mathbf{s}}_\ordpiv)\leq 
\Delta_{u}(s_j-{\tilde{s}_j})u(n), 
\end{equation} \normalsize
where $u(n)\triangleq \displaystyle \sum_{i=0}^{T-n} (\beta(1-p))^i$.
\end{lem}

\begin{proof}
See Appendix~\ref{app:lub}. 
\end{proof}
The result of Lemma~\ref{lem:upb} establishes that increasing the belief state 
of a node $j$ from ${\tilde{s}_j}$ to $s_j$ may increase the value of the pseudo 
value function, which is bounded by a linear function of the increase in the  
belief, $s_j-\tilde{s}_j$, and the function $u(n)$, which decreases with $n$ and 
corresponds to the maximum accumulated loss from TS $n$ to TS $T$.

\begin{lem}  \label{lem:swap} Consider two belief vectors $\mathbf{s}$ and 
$\mathbf{s}_\Pi$, such that permutation $\Pi$ is an $i,j$-swap, and $s_j\geq 
s_i$ for some $j\leq i$.  If $R(\cdot)$ is a bounded regular function, 
$\tau(\cdot)$ a monotonically increasing contracting map, and $ \beta \leq 1$, 
then 
\eqnsize\begin{equation}  \label{eq:swap}
 W_n(\mathbf{s})-W_n(\mathbf{s}_\Pi)\geq 0 \text{\normalsize{ if } } \Delta_{l} 
\geq \Delta_{u} \beta  p \frac{1-( \beta (1-p))^{T+1}}{1-\beta(1-p)}.
\end{equation} \normalsize
\end{lem}
\begin{proof}
See Appendix \ref{app:diff}.
\end{proof}

\begin{thm} \label{thm:1}If  $R(\cdot)$ is a bounded regular function, 
$\tau(\cdot)$ a monotonically increasing contracting map, $ \beta \leq 1$,  and 
$\Delta_{l} \geq \Delta_{u} \beta  p \frac{1-( \beta 
(1-p))^{T+1}}{1-\beta(1-p)}$, then MP is the optimal policy.
\end{thm}
\begin{proof}
The proof is done by backward induction. We have already shown that  MP is 
optimal at TS $T$. Then we assume that MP is optimal from TS $n+1$ until TS $T$, 
and we need to show that  MP is optimal at TS $n$. To show that  MP is optimal 
at TS $n$, using Lemma~\ref{lem:opt1}, we only need to show that (\ref{eq:lem1}) 
holds. This is proven in Lemma~\ref{lem:swap}, which completes the proof.
\end{proof}

The result of Theorem~\ref{thm:1} holds for any $R(\cdot)$ that is a bounded 
regular function. The reward function studied here, i.e., the sum expected throughput in (\ref{eq:reward}), is a bounded 
regular function, and we have $\Delta_{u}=\Delta_{l}=pB$. Finally, we state 
the optimality of MP for the EH problem studied in this section.

\begin{thm}\label{thm:2}
For the reward function $R(\cdot)$ defined in 
(\ref{eq:reward}), if the transition probabilities satisfy 
$p_{11} \geq p_{01}$ and $e_{0} \leq \frac{p_{10}}{p_{01}+p_{10}}$, then MP is 
the optimal policy.  
\end{thm}

\section{Simultaneous Energy Harvesting and Data Transmission with Batteryless 
Nodes} \label{sec:case2}
Now we consider another special case of the system model introduced in 
Section~\ref{sec:system_model}. We assume that the nodes cannot store energy, 
and the harvested energy is lost if not used immediately. This might apply to low-cost batteryless nodes. Energy 
available for transmission in TS~$n$ is equal to the energy harvested in 
TS~$n-1$, that is,  $B_i(n)=E_i^h(n-1)$. We denote by $s_i(n)$ the belief state 
of node~$i$ at TS~$n$, which is the expected energy available for transmission, 
that is, the probability that the node is in the harvesting state. The belief 
state transition probabilities are 
\eqnsize\begin{equation}    \label{eq:believe_state_case2}
 s_i(n+1)=\left \{ \begin{array}{lll}
\tau(s_i(n)) & \text{ if~}  i \notin \mathcal{K}^a(n), \\
p_{11} 	     & \text{ if~}  i \in \mathcal{K}^a(n) \text{ w.p.~} s_i(n),  \\
p_{01} 	     & \text{ if~}  i \in \mathcal{K}^a(n) \text{ w.p.~} 1-s_i(n), 
\end{array} \right .
\end{equation}\normalsize
where $\tau(s) = (p_{11} - p_{01} )s + p_{01}$, and since $p_{11}\leq p_{01}$, 
it is a monotonically increasing affine function. This implies that if $s_i \geq 
s_j$ then $\tau(s_i)\geq \tau(s_j)$, that is, the order of the idle nodes is 
preserved. We note that $i\in \mathcal{K}^a(n)$ with probability $p$, if $i\in 
\mathcal{K}(n)$.  The problem is to find a scheduling policy, $\mathcal{K}(n)$,  
such that the expected discounted sum throughput is 
maximised over a time horizon $T$. 

We define the pseudo value function as follows
\eqnsize\begin{equation} \label{eq:pseudo_val_case2}
\begin{array}{rcl}
W_n(\mathbf{s}_\Pi)&\triangleq&  R(\mathbf{s}_\Pi^{K}) 
+\beta\displaystyle\sum_{\mathclap{\!\!\!\mathcal{E} \subseteq 
\mathcal{S}_\Pi^{K}}} ~~\displaystyle \sum_{\mathclap{~~~l_\mathcal{E} \in  
\{0,1\}^{|\mathcal{E}|}}}   h(l_\mathcal{E},K) \\
&& \times W_{n+1} \left (\mathbf{P}_{11}\left (\Sigma l_\mathcal{E} \right ),   
\pmb{\tau}(\mathbf{s}_{\overline{\mathcal{E}}}),   
\mathbf{P}_{01}\left(\overline{\Sigma} l_\mathcal{E}\right ) \right ) ,\\
W_T(\mathbf{s}_\Pi)&\triangleq&  R(\mathbf{s}_\Pi^{K}) ,
\end{array}
\end{equation}  \normalsize
where we denote the set of active nodes  by $\mathcal{E}$ and the $i$th active 
node by  $\mathcal{E}(i)$. We define 
$l_{\mathcal{E}}=(l_{\mathcal{E}(1)},\ldots,l_{\mathcal{E}(|\mathcal{E}|)})$, 
such that $l_{\mathcal{E}(i)}=1$ if the EH process of the corresponding node is 
in the harvesting state, and $l_{\mathcal{E}(i)}=0$ otherwise. We define the 
function  $h(l_\mathcal{E},K)\triangleq q(|\mathcal{E}|,K) 
\displaystyle\prod_{j\in \mathcal{E}} s_j^{l_j} (1-s_j)^{(1-l_j)}$, where 
$q(|\mathcal{E}|,K)$ is defined in~(\ref{eq:kchannels}).  We denote by   
$\mathbf{P}_{01}(a)$ and $\mathbf{P}_{11}(a)$ the vectors $(p_{01}, \ldots, 
p_{01})$ and $(p_{11}, \ldots, p_{11})$, respectively, of length $a$, and  we 
define  $\Sigma l_\mathcal{E}\triangleq\displaystyle \sum_{i \in \mathcal{E}} 
l_i$, and $\overline{\Sigma} l_\mathcal{E}\triangleq|\mathcal{E}|-\displaystyle 
\sum_{i \in \mathcal{E}} l_i$. The operator  $\pmb{\tau}(\cdot)$ applies the 
mapping $\tau(\cdot)$ to all its components. The pseudo value function schedules 
the nodes according to permutation $\Pi$, and if $\mathbf{s}_\Pi$ is ordered, then 
(\ref{eq:pseudo_val_case2}) is the \mbox{value function of MP.} 

Swapping the order of two scheduled nodes does not change 
the value of the pseudo value function, that is, the pseudo value function is 
symmetric. This property is similar to that in (\ref{eq:Wsymmetric}), but only 
for $i,j \leq K$. Similarly to \cite{RMAB:Ahmad2009} and \cite{RMAB:Ahmad2009a}, 
the mapping $\tau(\cdot)$ is linear, and hence, the pseudo value function is 
affine in each of its elements. This implies that, if $\Pi$ is an $i,j$-swap of 
$\barpi$, then
\begin{IEEEeqnarray}[\normalsize]{rlL} \label{eq:pseudo_linprop}
W&_n&(\mathbf{s}_\Pi)-W_n(\mathbf{s}_\barpi) \nonumber \\
&=&(s_{\Pi(j)} -s_{\Pi(i)}) \Big( W_n(\ldots, s_{\Pi(j)} =1, \ldots,  s_{\Pi(i)} 
=0, \ldots)\nonumber \\ &&-W_n(\ldots, s_{\Pi(j)} =0, \ldots,  s_{\Pi(i)} =1, 
\ldots) \Big).
\end{IEEEeqnarray}

MP schedules the nodes whose EH processes are 
more likely to be in the harvesting state.  Initially, nodes are 
ordered according to an initial belief. If a node is active,  it is sent to the 
first position of the queue if it is in the harvesting state, and to the last 
position if it is in the non-harvesting state. The idle nodes are moved forward 
in the queue. Due to the monotonicity of $\tau(\cdot)$, MP continues scheduling 
a node until it is active and its EH process is in the non-harvesting state.

\subsection{Proof of the optimality of MP}
We note that the result of Lemma~\ref{lem:opt1} is applicable in this case. If 
Lemma \ref{lem:case2} holds, the same arguments as in Theorem \ref{thm:1} can be 
used to prove the optimality of MP.

\begin{lem} \label{lem:case2} Let $\Pi$ be an $i,j$-swap, and consider a 
permutation $\barpi$, such that  $\barpi(k)=k-1$, for $\forall k \neq 1$ and 
$\barpi(1)=N$. If $s_j \geq s_i$ for some $j\leq i$, then we have the 
inequalities 
 \begin{subequations}\label{eq:lem:case2}
  \begin{IEEEeqnarray}[\eqnsize]{rcl}
1+W_n(\mathbf{s}_{\barpi}) &\geq & W_n(\mathbf{s}),\label{eq:lem:case2:1}\\ 
W_n(\mathbf{s}) &\geq& W_n(\mathbf{s}_{\Pi}). \label{eq:lem:case2:2}
  \end{IEEEeqnarray}
\end{subequations}
\end{lem}
\begin{proof}
 The proof follows from the similar arguments as in \cite{RMAB:Ahmad2009a}. In 
particular, we use backward induction in (\ref{eq:lem:case2:1}) and 
(\ref{eq:lem:case2:2}), and a sample-path argument. A sketch of the proof is 
provided in Appendix~\ref{app:case2}.
\end{proof}
Note that (\ref{eq:lem:case2:1}) and (\ref{eq:lem:case2:2}) are similar to 
(\ref{eq:upb}) and (\ref{eq:swap}), respectively.

\begin{thm}\label{thm:3} If the reward function is $R(\mathcal{K}(n))=p\!\!\! 
\displaystyle\sum_{i\in\mathcal{K}(n)}\!\!\!s_i(n)$, and $p_{11} \geq p_{01}$,  MP is the optimal policy.    
\end{thm}

\begin{proof}
Theorem \ref{thm:3} can be proven by using the same arguments as in Theorem 
\ref{thm:1} \mbox{and Lemmas~\ref{lem:opt1}  and \ref{lem:case2}.} 
\end{proof}

\begin{rem}
 This problem is similar to the opportunistic multi-channel access problem 
studied in \cite{RMAB:Ahmad2009,RMAB:Ahmad2009a,RMAB:Mansourifard2012,RMAB:Wang2012}, with 
imperfect channel sensing, such that, at each attempt, a channel can not be 
sensed with probability $1-p$, independent of its channel state. While the MP 
has been proven to be optimal in the case of perfect channel sensing, i.e., 
$p=1$, \cite{RMAB:Ahmad2009a}, the case with sensing errors, i.e., $p\neq1$, has 
not been considered in the literature. We also note that this model of imperfect 
channel detection is different from that in \cite{RMAB:Liu2010}. 
\end{rem}

\begin{rem}
 Using similar techniques as in \cite{RMAB:Ahmad2009} the MP optimality results 
of Sections~\ref{sec:case1} and \ref{sec:case2} can be extended from the finite 
horizon discounted reward criteria to the infinite horizon with
discounted reward, and to the infinite horizon with average reward criteria.
\end{rem}

\section{Upper Bound on the Performance of the Optimal Scheduling Policy} 
\label{sec:upp}

Next we derive an upper bound on the performance of the optimal 
policy for the general model in Section~III under the average reward 
criteria and infinite time horizon. The RMAB problem with an infinite horizon 
discounted reward criteria is studied in \cite{RMAB:D.Betsimas2000}, and it is 
shown that an upper bound can be computed in polynomial time using LP.

The decision of scheduling a node in a TS affects the scheduling of the other 
nodes in the same TS, since exactly $K$ nodes have to be scheduled at each TS.  
Whittle \cite{RMAB:Whittle1988} proposed to relax the original problem 
constraint, and impose instead that the number of nodes that are scheduled at 
each TS is $K$ \emph{on average}. In the relaxed problem, since the nodes are 
symmetric, one can decouple the original RMAB problem into $N$ RMAB problems, 
one for each node.  As before, we denote  by $s=(l,h)\in \mathcal{W}$ the belief 
state of a node, where $l$ is the number of TSs the node has been idle for, and 
$h$ the EH state last time the node was scheduled, and $\mathcal{W}$ the belief 
state space.  We denote by  $\pi(s)$ the probability that a node is scheduled if 
it is in state $s$, by $p(s)$ the steady state probability of state $s$, and by 
$p_{\tilde{s},s}(a)$ the state transition probability function from state 
$\tilde{s}$ to $s$ if action $a\in\{0,1\}$ is taken, where $a=1$ if the node is 
scheduled in this TS, and $a=0$, otherwise.  The optimization problem is 
\eqnsize
\begin{equation}   \label{eq:opt_problemv_rel1}
\begin{aligned} 
\max_{\pi(s),p(s)} &\sum_{s\in\mathcal{W}} R(s) \pi(s) p(s) \\
\text{\normalsize{s.t.} }  &
	      p(s)= \sum_{\tilde{s}\in \mathcal{W}} p(\tilde{s}) 
[(1-\pi(s))p_{\tilde{s},s}(0)+\pi(s) p_{\tilde{s},s}(1)], \\
	    & \sum_{s\in \mathcal{W}}\pi(s) p(s) =\frac{K}{N}, 
\text{\normalsize{ and }} \sum_{s\in \mathcal{W}} p(s) =1,\\
	    \end{aligned}
\end{equation} \normalsize
where $0 \leq \pi(s),~p(s) \leq 1$, and $R(s)$ is the expected throughput of a 
node if it is in state $s$. Note that the node is scheduled every $\frac{N}{K}$ 
TSs on average. This implies that, for $p=1$, the maximum time a node can be idle is 
finite, and hence, the state space $\mathcal{W}$ is finite. If $p\neq1$, one can truncate the state space by bounding the maximum time a 
node can be idle, i.e., imposing that $l$ is bounded.  The 
problem~(\ref{eq:opt_problemv_rel1}) has a linear objective function and 
linear constrains, and the state space is finite, therefore it can be solved in polynomial 
time with LP.

\section{Numerical Results} \label{sec:numerical}

In this section we numerically study the performances of different scheduling policies for 
the general case described in Section \ref{sec:system_model}. In particular, we consider MP 
which is optimal for the cases studied in Sections \ref{sec:case1} and 
\ref{sec:case2}, the RR policy, which schedules the nodes in a cyclic fashion 
according to an initial random order, and a random policy, which at each TS 
schedules $K$ random nodes regardless of the history. We measure the performance 
of the scheduling policies as the average throughput per TS over a time horizon 
of $T=1000$, that is, we consider $\beta=1$ and normalise 
\eqref{eq:opt_problemv1} by $T$. We perform $100$ repetitions for each 
experiment and average the results. We assume, unless otherwise stated, a total 
of $N=30$ EH nodes, $K=5$ available channels, and a probability $p=0.5$ for a 
node to be operative in each TS. We assume that all the nodes and EH processes 
are symmetric, the batteries have a capacity of $B=5$ energy units, and the 
transition probabilities of the EH processes are $p_{11}=p_{00}=0.9$. Notice 
that, on average, each node is scheduled every $\frac{N}{K}$ TSs. Hence, if 
$\frac{N}{K}$ is large the nodes  remain idle for larger periods. This implies 
that when $\frac{N}{K}$ is large, since the nodes harvest over many TSs without 
being scheduled, there are more energy overflows in the system. In the numerical 
results we have included the infinite horizon upper bound of 
Section~\ref{sec:upp}, which for large $T$ is a tight upper bound on the 
finite horizon case.

In Figure~\ref{fig:NvsC} we investigate the impact of the number of nodes on the 
 throughput, when the number of available channels, $K$,  is fixed. The 
throughput increases with the number of nodes, and due to the battery 
overflows, saturates when the number of nodes is large. By increasing the 
battery capacity, hence reducing the battery overflows, the throughput saturates 
with a higher number of nodes and at a higher value. We observe that MP has a 
performance close to that of the upper bound, the random policy has a lower 
performance than the others; and the gap between different curves increases with the 
battery capacity.

In Figure~\ref{fig:bat} we investigate the effect of the battery capacity, $B$, 
on the system throughput when the number of nodes is fixed. Clearly, the larger 
the battery capacity the fewer battery overflows will occur. The throughput 
increases with the battery capacity, and due to the limited amount of energy 
that the nodes can harvest, it saturates at a certain value. By increasing the 
number of  available channels, $K$,  which also reduces the battery overflow, 
the throughput saturates more quickly as a function of the battery capacity, but 
at higher values. The performances of the scheduling policies are similar to 
those observed in Figure~\ref{fig:NvsC}. 

\begin{figure}[h!]
     \begin{center}
       \subfigure[$N/K$ ]{%
         \includegraphics[width=0.45\textwidth]{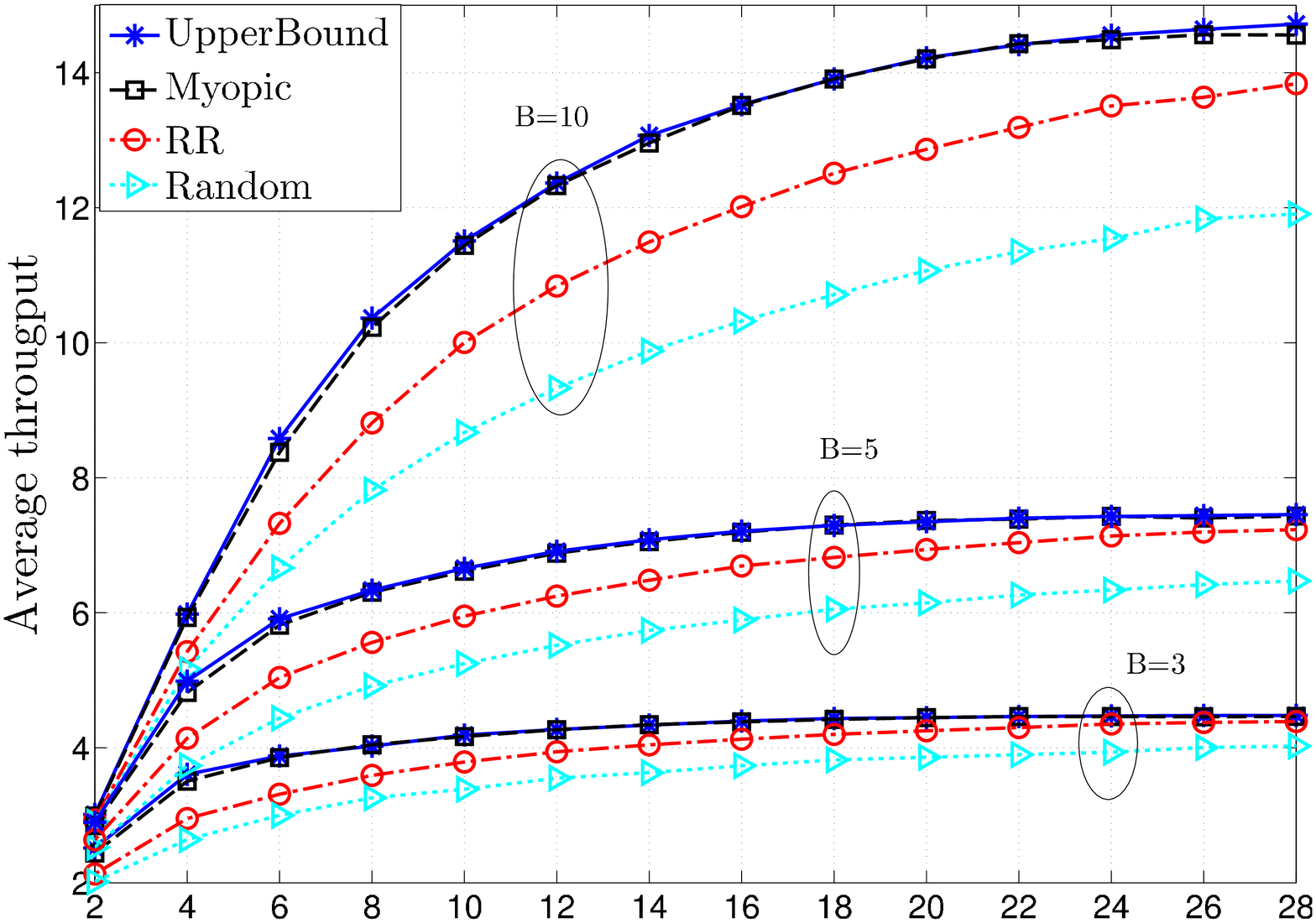} \label{fig:NvsC} 
        }%
        \\
        \subfigure[Battery capacity ($B$)  ]{%
\includegraphics[width=0.45\textwidth]{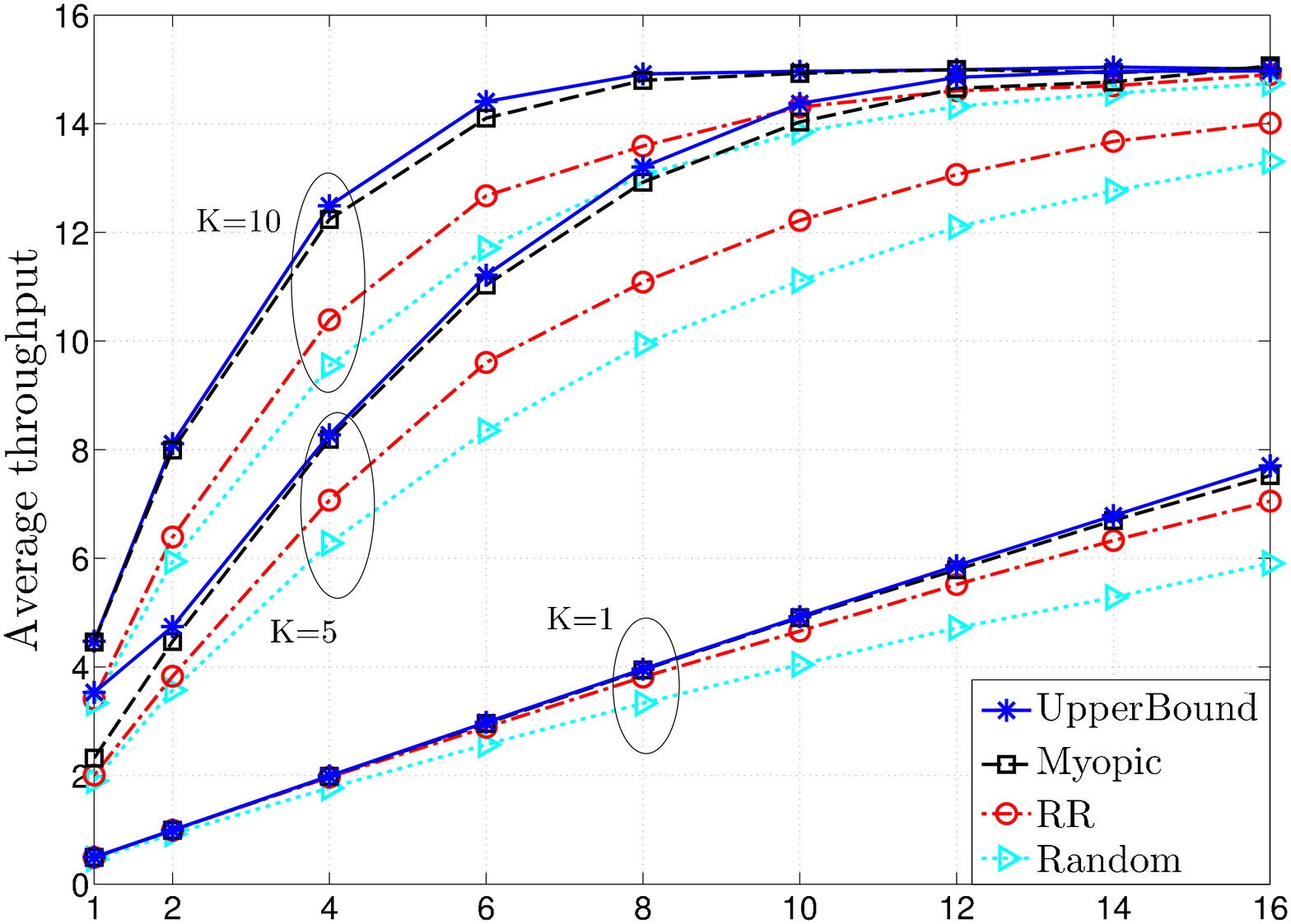}\label{fig:bat}
       }
    \end{center} 
    \caption{ (a) Average throughput vs. number of nodes, $N$, with $K=5$ 
channels, and battery capacity $B=3,5,10$, and (b) average throughput vs. 
battery capacity, $N=30$, and $K=1,5,10$.} 
\label{fig:param}
\end{figure}

 \begin{figure}[h!]
     \begin{center}
       \subfigure[$p_{11}$  ($p_{00}=0.5$)]{%
           \label{fig:p11p00=05}
\includegraphics[width=0.45\textwidth]{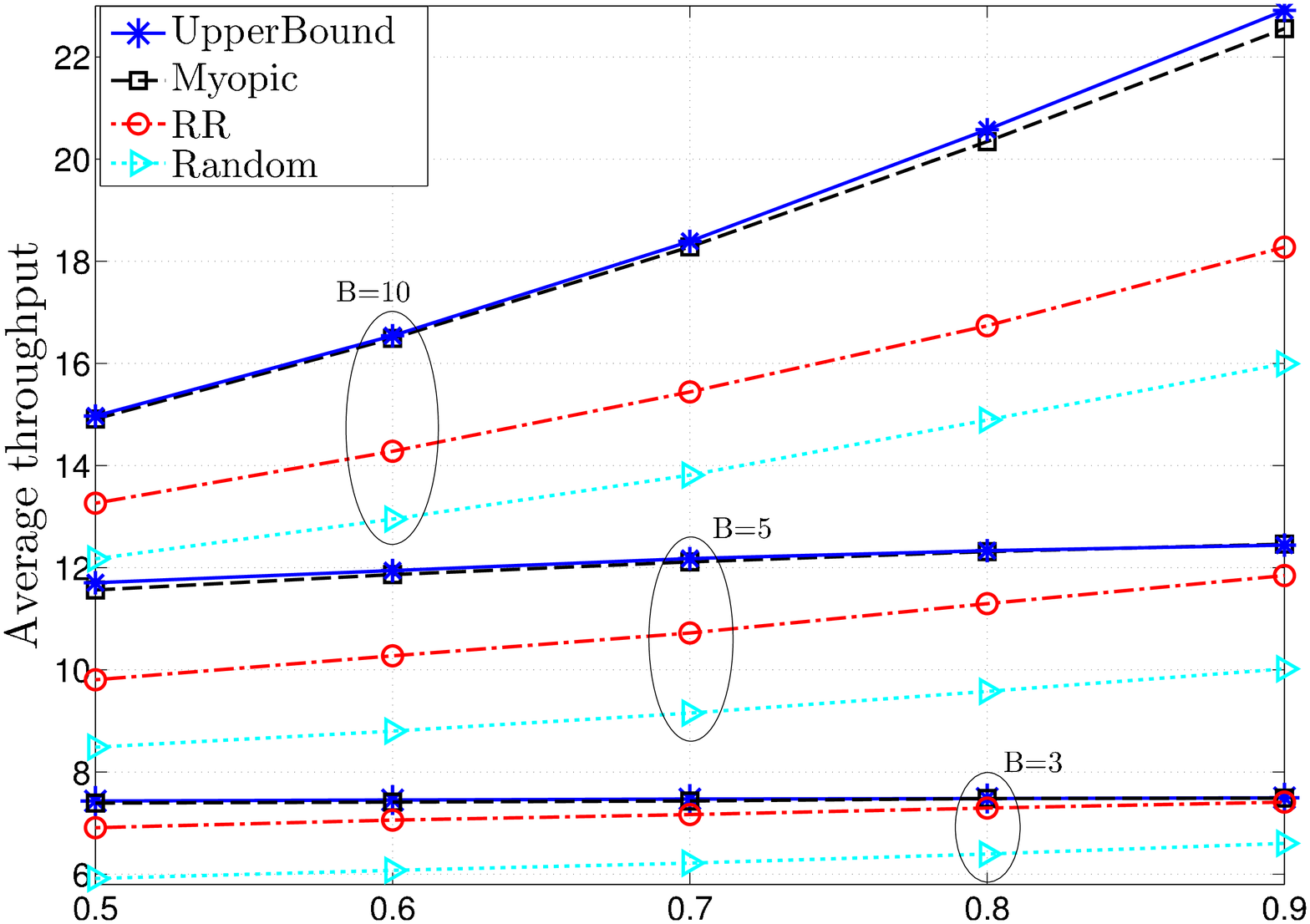}
        }%
        \\
        \subfigure[$p_{11}$  ($p_{00}=0.9$)]{%
           \label{fig:p11p00=09}
\includegraphics[width=0.45\textwidth]{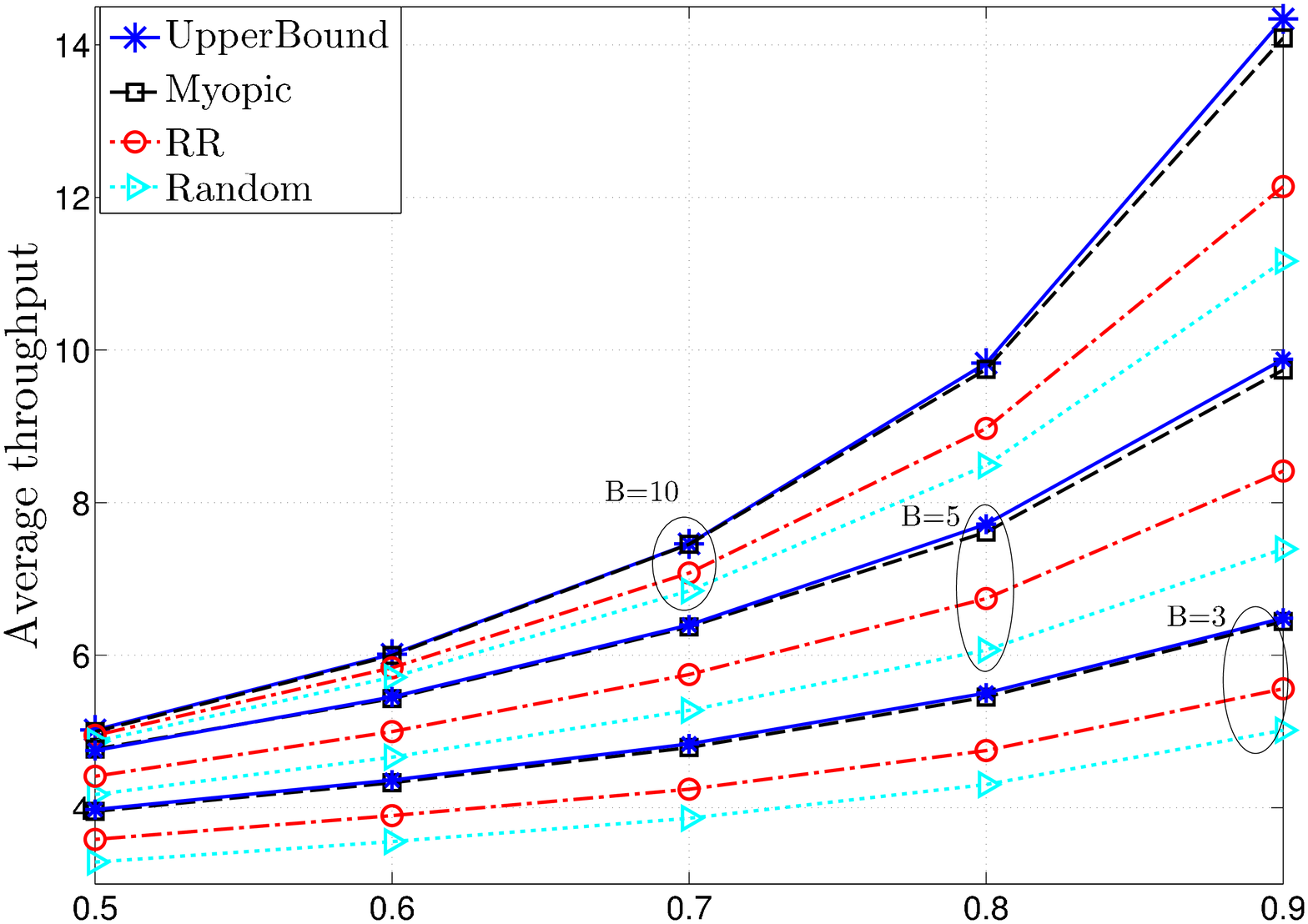}
       }
    \end{center}
    \caption{Average throughput for different EH process transition 
probabilities, $N=30$, $K=5$, and $B=3,5,10$.} \label{fig:prob} 
\end{figure}

Figure~\ref{fig:prob} shows the average throughput for different EH process 
transition probabilities. We note that the amount of energy  arriving to the 
system increases with $p_{11}$ and decreases with $p_{00}$. As expected, we 
observe in Figure~\ref{fig:prob} that the throughput increases with $p_{11}$, and the values in  Figure~\ref{fig:p11p00=05} are notably higher than those 
in Figure~\ref{fig:p11p00=09}. MP is a policy which maximises the immediate 
throughput at each TS, and does not take into account the future TSs. We observe 
in Figure~\ref{fig:p11p00=09} for $B=\{5,10\}$ and in  
Figure~\ref{fig:p11p00=05} for $B=10$ that,  if the EH state has low correlation 
across TSs, that is, $p_{11}=\{0.5, 0.6\}$, the throughput obtained by MP is 
similar to that of the upper bound.  On the contrary, if it has high correlation 
across TSs, that is  $p_{11}=\{0.8, 0.9\}$, the throughput falls below the upper 
bound.  This is due to the fact that when the state transitions have low 
correlation it is difficult to reliably predict the impact of the actions on the 
future rewards, and no transmission strategy can improve upon MP. Our numerical 
results indicate, that even in scenarios in which the MP cannot be shown to be 
theoretically optimal, it performs very close to the upper bound, obtained for 
an infinite horizon problem. 

\section{Conclusions} \label{sec:con}
We have studied a scheduling problem in a multi-access communication system with 
EH nodes, in which the harvested energy at each node is modeled as a Markov 
process. We have modeled the system as an RMAB problem, and shown the 
optimality of MP in two settings: i) when the nodes cannot harvest energy and 
transmit simultaneously and the EH process state is independent of the past 
states after a node is active; ii) when the nodes have no battery. The results of this paper suggest that although the optimal scheduling in large EH networks requires high computational complexity, in some cases there exist simple and practical scheduling policies that have almost optimal performance. This can have an impact on the design of scheduling 
policies for large low-power wireless sensor networks equipped with energy 
harvesting devices and limited storage.  


\begin{appendices}
\section{}\label{app:contracting_map}
We denote the probability that the battery of a node is not full if the node has 
been idle for the last $n$ TSs  by $p_{nf}(n)$. It is easy to note that 
$p_{nf}(n)$ is a decreasing function of $n$. If the node has been idle for $n$ TSs, 
we denote the probability of the EH process being in state $0$ and $1$, by 
$p_{0}(n)\triangleq p_{10}+p_{0}(n-1)(p_{11}-p_{01})$ and  
$p_1(n)\triangleq1-p_0(n)$, respectively. We set $p_0(0)=e_0$. Since $p_{11}\geq 
p_{01}$ and $e_0 \leq  \frac{p_{10}}{p_{01}+p_{10}}$, $p_{0}(n)$ monotonically 
increases to the steady state distribution~(\cite[Appendix~B]{RMAB:Zhao2008}). 

We denote the belief state of a node that has been idle for $n$ TSs by $z_n$. If 
the node has been idle for $n+1$ TSs, the expected battery level is 
$z_{n+1}=\tau(z_n)=z_n+\frac{p_{nf}(n)}{B}  (p_{01}p_{0}(n)+p_{11}p_{1}(n))$, 
which is a monotonically increasing function. If $n\geq m$, then $z_n\geq z_m$ 
and $\tau(z_n) \geq \tau(z_m)$. By applying the definition of $p_{1}(n)$, we get 
 $z_{n+1}=z_n+\frac{p_{nf}(n)}{B}  (p_{11}-p_{0}(n)(p_{11}-p_{00}))$. If we 
assume that $n\geq m$, we have     
\begin{IEEEeqnarray}[\small]{rCL}
 \|\tau(z_n)-\tau(z_m) \| &=&z_n-z_m +\frac{p_{nf}(n)}{B} (p_{11} 
-p_0(n)(p_{11}-p_{01})) \nonumber\\ 
 &&-\frac{p_{nf}(m)}{B} (p_{11}\!-\!p_0(m)(p_{11}-p_{01})) \nonumber\\
  &\leq& z_n\!-\!z_m\!-\!\frac{p_{nf}(n)}{B} (p_{11}\!-\!p_{01}) 
(p_0(n)\!-\!p_0(m))\nonumber\\
  &\leq& z_n-z_m  \nonumber,
\end{IEEEeqnarray}
where the first inequality follows since $p_{nf}(n) \leq p_{nf}(m)$, and the 
second inequality follows since $p_{0}(n)$ is monotonically increasing and 
$p_{11} \geq p_{01}$.

\section{}\label{app:lub}
 The proof uses backward induction. We denote by $\mathcal{S}_\ordpiv^{K}$ and 
$\tilde{\mathcal{S}}_\ordpiv^{K}$ the nodes scheduled from $\mathbf{s}_\ordpiv$ 
and $\tilde{\mathbf{s}}_\ordpiv$, respectively. We first observe that 
(\ref{eq:upb}) holds for $n=T$.  This follows from the bounded regularity of 
$R(\cdot)$, noting that $u(T)=1$, and distinguishing four possible cases.

\begin{itemize}
  \item Case $1$: $j \in \mathcal{S}_\ordpiv^{K}$ and $j \in 
\tilde{\mathcal{S}}_\ordpiv^{K}$, i.e., node $j$ is scheduled in both cases.
  \begin{IEEEeqnarray}[\small]{rcl}
  W&_T&(\mathbf{s}_\ordpiv)- W_T(\tilde{\mathbf{s}}_\ordpiv) \nonumber \\
  &=& 
R(s_{\ordpiv(1)},\ldots,s_j,\ldots,s_{\ordpiv(K)})-R(\tilde{s}_{\ordpiv(1)},
\ldots,\tilde{s}_j,\ldots,\tilde{s}_{\ordpiv(K)})\nonumber \\
  &=&  s_j 
R(s_{\ordpiv(1)},\ldots,1,\ldots,s_{\ordpiv(K)})+(1-s_j)R(s_{\ordpiv(1)},\ldots,
0,\nonumber \\
  && \ldots,s_{\ordpiv(K)}) 
-\tilde{s}_jR(\tilde{s}_{\ordpiv(1)},\ldots,1,\ldots,\tilde{s}_{\ordpiv(K)}
)\nonumber \\
  &&-(1-\tilde{s}_j) 
R(\tilde{s}_{\ordpiv(1)},\ldots,0,\ldots,\tilde{s}_{\ordpiv(K)}) \nonumber\\
  &=&(s_j-\tilde{s}_j) 
\big(R(s_{\ordpiv(1)},\ldots,1,\ldots,s_{\ordpiv(K)})\nonumber \\
  &&-R(s_{\ordpiv(1)},\ldots,0,\ldots,s_{\ordpiv(K)}) \big) \nonumber \\
  &\leq& (s_j-{\tilde{s}_j})\Delta_{u} u(T),  \nonumber 
  \end{IEEEeqnarray}
   where the second equality follows from the decomposability of $R(\cdot)$. 
Since $R(\cdot)$ is symmetric and the belief vectors are equal but for node $j$, 
we have 
$R(s_{\ordpiv(1)},\ldots,\tilde{s}_j=k,\ldots,s_{\ordpiv(K)})=R(\tilde{s}_{
\ordpiv(1)},\ldots,\tilde{s}_j=k,\ldots,\tilde{s}_{\ordpiv(N)})$, which we use 
in the third equality. Finally, the inequality follows from the boundedness of 
$R(\cdot)$. 
  \item Case $2$: $j \notin \mathcal{S}_\ordpiv^{K}$ and $j \notin 
\tilde{\mathcal{S}}_\ordpiv^{K}$,  i.e., node $j$ is not scheduled in either 
case. The same nodes with the same beliefs are scheduled in both cases, hence, 
$\mathbf{s}_\ordpiv^{K}=\tilde{\mathbf{s}}_\ordpiv^{K}$, and $ 
W_T(\mathbf{s}_\ordpiv)- W_T(\tilde{\mathbf{s}}_\ordpiv)=0$.
  \item Case $3$:  $j \in \mathcal{S}_\ordpiv^{K}$ and $j \notin 
\tilde{\mathcal{S}}_\ordpiv^{K}$. In this case there exists a node $m \in 
\tilde{\mathcal{S}}_\ordpiv^{K}$  such that $s_j \geq s_m \geq \tilde{s}_j$, and 
$m \notin \mathcal{S}_\ordpiv^{K}$
      \begin{subequations}
  \begin{IEEEeqnarray}[\small]{lll}
 W&_T&(\mathbf{s}_\ordpiv)- W_T(\tilde{\mathbf{s}}_\ordpiv) \nonumber \\ 
 &=& (s_j-s_m)\big(R(s_{\ordpiv(1)},\ldots,1,\ldots,s_{\ordpiv(K)})\nonumber \\
 &&-R(s_{\ordpiv(1)},\ldots,0,\ldots,s_{\ordpiv(K)}) \big)  \nonumber \\
  &\leq& (s_j-{\tilde{s}_j}) 
\big(R(s_{\ordpiv(1)},\ldots,1,\ldots,s_{\ordpiv(K)})\nonumber \\
  &&-R(s_{\ordpiv(1)},\ldots,0,\ldots,s_{\ordpiv(K)}) \big) \nonumber\\
  &\leq& (s_j-{\tilde{s}_j}) \Delta_{u} u(T), \nonumber  
  \end{IEEEeqnarray}
  \end{subequations}
  where the first equality follows similar to Case $1$, the second equality from 
the fact that $s_m \geq \tilde{s}_j $, and the last inequality from the 
boundedness of $R(\cdot)$. Note that node $m$ is the node with the highest 
belief state that is not scheduled in $W_T(\mathbf{s}_\ordpiv)$, and the node 
with the lowest belief state scheduled~in~$ W_T(\tilde{\mathbf{s}}_\ordpiv)$.
\item Case $4$:  $j \notin \mathcal{S}_\ordpiv^{K}$ and $j \in 
\tilde{\mathcal{S}}_\ordpiv^{K}$. This case is not possible since the vectors 
$\mathbf{s}_\ordpiv$ and $\tilde{\mathbf{s}}_\ordpiv$ are ordered and $s_j \geq 
{\tilde{s}_j}$, hence, if ${\tilde{s}_j}$ is scheduled then $s_j$ must be 
scheduled too.
\end{itemize}
Now, we assume that (\ref{eq:upb}) holds from TS $n+1$ up to $T$, and show 
that it holds for TS $n$ as well. We distinguish three cases:
\begin{itemize}
  \item Case $1$: $j \in \mathcal{S}_\ordpiv^{K}$ and $j \in 
\tilde{\mathcal{S}}_\ordpiv^{K}$ in  \eqref{seq:upb3_all}, i.e., node $j$ is 
scheduled in both cases.  The first and second summations in the first line of (\ref{seq:upb3_1})   
correspond to the cases in which node $j\in  \mathcal{S}_\ordpiv^{K} $ is idle 
and active, respectively, in TS $n$. Similarly, first and second summations in 
the second line of (\ref{seq:upb3_1})  correspond to the cases in which node 
$j\in  \tilde{\mathcal{S}}_\ordpiv^{K}$ is idle and active, respectively, in TS 
$n$. Note that the belief state vector 
$\tilde{\mathbf{s}}_{\overline{\mathcal{E}\cup j}}$ includes the belief states 
of all the nodes in $\tilde{\mathbf{s}}_\ordpiv$, but those in $\mathcal{E}$ and 
$\tilde{s}_j$, hence, it is equivalent to the belief state vector 
$\mathbf{s}_{\overline{\mathcal{E}\cup {j}}}$. We use this fact to get 
(\ref{seq:upb3_3}). Note that the belief state vectors in  (\ref{seq:upb3_3}) 
differ only in the belief states of node $j$, namely, $\tau(s_j)$ and 
$\tau(\tilde{s}_j)$ are the beliefs of node $j$ in vectors 
$[\mathrm{T}(\mathbf{s}_{\overline{\mathcal{E}}}),\mathbf{0}(|\mathcal{E}|)]$ 
and  
$[\mathrm{T}(\tilde{\mathbf{s}}_{\overline{\mathcal{E}}}),\mathbf{0}(|\mathcal{E
}|)]$, respectively; and hence, we use the induction hypothesis in the summation 
of (\ref{seq:upb3_3}) to obtain (\ref{seq:upb3_4}). The summation in 
(\ref{seq:upb3_4}) is over all possible operative/inoperative combinations of 
the nodes in $ \mathcal{S}_\ordpiv^{K}\backslash \{j\}$, and it is equal to one. 
This fact together with the boundedness and the decomposability of $R(\cdot)$ 
are used in (\ref{seq:upb3_4}) to get (\ref{seq:upb3_4b}). The contracting 
property of $\tau(\cdot)$, and  the definition of $u(n)$ are used in 
(\ref{seq:upb3_4c}) and (\ref{seq:upb3_5}), respectively.
   \begin{floateq}
    \begin{subequations} \label{seq:upb3_all}
    \begin{IEEEeqnarray}[\normalsize]{rCL} 
   &W_n&(\mathbf{s}_\ordpiv) - W_n(\tilde{\mathbf{s}}_\ordpiv) \nonumber   \\ 
   &=& R(\mathbf{s}^K_\ordpiv) + (1\!-\!p) \beta \displaystyle 
\!\!\!\!\!\!\!\!\! \sum_{\mathcal{E}\subseteq 
\mathcal{S}_\ordpiv^{K}\backslash\{j\}}   \!\!\!\!\!\!\! 
\!q(|\mathcal{E}|,\!K\!\!-\!\!1) W_{n+1} ( 
[T(\mathbf{s}_{\overline{\mathcal{E}}}), \mathbf{0}(|\mathcal{E}|) ]) + p\beta 
\displaystyle \!\!\!\!\!\!\!\!\! \sum_{\mathcal{E}\subseteq 
\mathcal{S}_\ordpiv^{K}\backslash\{j\}}  \!\!\!\!\!\!\! 
\!\!q(|\mathcal{E}|,\!K\!\!-\!\!1) W_{n+1} ( 
[T(\mathbf{s}_{\overline{\mathcal{E}\cup j}}), 
\mathbf{0}(|\mathcal{E}|\!\!+\!\!1) ]) \nonumber \\
   &-& R(\tilde{\mathbf{s}}^K_\ordpiv)\! - \! (1\!-\!p) \beta \displaystyle 
\!\!\!\!\!\!\!\!\! \sum_{\mathcal{E}\subseteq 
\tilde{\mathcal{S}}_\ordpiv^{K}\backslash\{j\}}  \!\!\!\!\!\!\! 
\!\!q(|\mathcal{E}|,\!K\!\!-\!\!1) W_{n+1} ( 
[T(\tilde{\mathbf{s}}_{\overline{\mathcal{E}}}), \mathbf{0}(|\mathcal{E}|) ])\! 
-p\beta  \displaystyle \!\!\!\!\!\!\!\!\! \sum_{\mathcal{E}\subseteq 
\tilde{\mathcal{S}}_\ordpiv^{K}\backslash\{j\}}  \!\!\!\!\!\!\! 
\!\!q(|\mathcal{E}|,\!K\!\!-\!\!1) W_{n+1} ( 
[T(\tilde{\mathbf{s}}_{\overline{\mathcal{E}\cup {j}}}), 
\mathbf{0}(|\mathcal{E}|\!\!+\!\!1) ]) \label{seq:upb3_1}\\
  &=&  
R(\mathbf{s}^K_\ordpiv)-R(\tilde{\mathbf{s}}
^K_\ordpiv)+(1\!-\!p)\beta\displaystyle \!\!\!\!\!\!\!\!\! 
\sum_{\mathcal{E}\subseteq \mathcal{S}_\ordpiv^{K}\backslash\{j\}} 
\!\!\!\!\!\!\! \!\!q(|\mathcal{E}|,\!K\!\!-\!\!1)   \Big ( 
W_{n+1}([\mathrm{T}(\mathbf{s}_{\overline{\mathcal{E}}}),\mathbf{0}(|\mathcal{E}
|)]) 
-W_{n+1}([\mathrm{T}(\tilde{\mathbf{s}}_{\overline{\mathcal{E}}}),\mathbf{0}
(|\mathcal{E}|)]) \Big ) \label{seq:upb3_3}\\
  &\leq& R(\mathbf{s}^K_\ordpiv)-R(\tilde{\mathbf{s}}^K_\ordpiv)+ 
(1\!-\!p)\beta\displaystyle\!\!\!\!\!\! \sum_{\mathcal{E}\subseteq 
\mathcal{S}_\ordpiv^{K}\backslash\{j\}} \!\!\!\! 
\!\!q(|\mathcal{E}|,\!K\!\!-\!\!1) \Big ( 
\Delta_{u}(\tau(s_j)-\tau(\tilde{s}_j))u(n+1)\Big) \label{seq:upb3_4}\\
  &\leq& \Delta_{u} (s_j-\tilde{s}_j)+ (1\!-\!p)\beta 
\Delta_{u}(\tau(s_j)-\tau(\tilde{s}_j))u(n+1)\label{seq:upb3_4b}\\
&\leq& \Delta_{u} (s_j-\tilde{s}_j)+ (1\!-\!p)\beta  
\Delta_{u}(s_j-\tilde{s}_j)u(n+1)\label{seq:upb3_4c}\\
&\leq& \Delta_{u} (s_j-\tilde{s}_j)\Big( 1+ \beta(1\!-\!p) \displaystyle\!\!\!\! 
\sum_{i=0}^{T-n-1} (\beta(1\!-\!p))^i\Big) \label{seq:upb3_5}\\
&=& \Delta_{u} (s_j-\tilde{s}_j)u(n),
\end{IEEEeqnarray}
\end{subequations}
\end{floateq}

\item Case $2$:$j \notin \mathcal{S}_\ordpiv^{K}$ and $j \notin 
\tilde{\mathcal{S}}_\ordpiv^{K}$, i.e., the same nodes are scheduled from 
$\mathbf{s}_\ordpiv$ and $\tilde{\mathbf{s}}_\ordpiv$,  and node $j$ is not 
scheduled in either case. Then
   \begin{subequations}
    \begin{IEEEeqnarray}[\normalsize]{rCL}
  W_n(\mathbf{s}_\ordpiv)&- &W_n(\tilde{\mathbf{s}}_\ordpiv) \nonumber\\
  &=&\beta\!\!\!\displaystyle\sum_{\mathcal{E}\subseteq\mathcal{S}_\ordpiv^{K}}   
\!\!\!q(|\mathcal{E}|,K)\Big ( W_{n+1}([\mathrm{T}( 
\mathbf{s}_{\overline{\mathcal{E}}}), \mathbf{0}(|\mathcal{E}|)]) \nonumber\\
&&-W_{n+1}([\mathrm{T}( \tilde{\mathbf{s}}_{\overline{\mathcal{E}}})
\mathbf{0}(|\mathcal{E}|)])\Big)\label{seq:upb4_1}\\
   %
   %
   &\leq&  \Delta_{u}(s_j -\tilde{s}_j)  \beta  u(n+1)  \label{seq:upb4_2}\\
   &\leq&  \Delta_{u}(s_j -\tilde{s}_j)  \beta \!\!\! 
\displaystyle\sum_{i=0}^{T-n-1} (\beta(1-p))^i  \label{seq:upb4_3}\\
   %
   %
   %
   &\leq& \Delta_{u} (s_j -\tilde{s}_j)u(n),  \label{seq:upb4_4}
  \end{IEEEeqnarray}
    \end{subequations}
    
  where (\ref{seq:upb4_1}) follows since the value of the expected immediate 
rewards in TS $n$ are the same. The belief state vectors at TS $n+1$ are equal 
but for the belief state of node $j$, that is, $\tau(s_j)$ and 
$\tau(\tilde{s}_j)$ are the beliefs of node $j$ in $\mathrm{T}( 
\mathbf{s}_{\overline{\mathcal{E}}})$ and  
$\mathrm{T}(\tilde{\mathbf{s}}_{\overline{\mathcal{E}}})$, respectively. In 
(\ref{seq:upb4_1}), similarly to (\ref{seq:upb3_4}), (\ref{seq:upb3_4b}), and 
(\ref{seq:upb3_4c}), we apply the induction hypothesis, the contracting map 
property, and the fact that the summation is equal to one,  to obtain  
(\ref{seq:upb4_2}).  We use $\beta \leq 1$ and the definition of $u(n)$ to 
obtain (\ref{seq:upb4_3}) and \eqref{seq:upb4_4}, respectively. 
  
  \item Case $3$: $j \in \mathcal{S}_\ordpiv^{K}$ and $j \notin 
\tilde{\mathcal{S}}_\ordpiv^{K}$ in \eqref{seq:upb5_all}, i.e., there exists  $m\in 
\tilde{\mathcal{S}}_\ordpiv^{K}$  such that $s_j \geq s_m=\tilde{s}_m \geq 
\tilde{s}_j$ and that $m\notin \mathcal{S}_\ordpiv^{K}$. Hence, 
$\mathcal{S}_\ordpiv^{K}$ and $ \tilde{\mathcal{S}}_\ordpiv^{K}$ differ only in 
one element. To obtain (\ref{seq:upb5_2}) we use the symmetry property of the pseudo 
value function and the fact that the belief vectors are equal but for node $j$; 
in  (\ref{seq:upb5_3}) we add and subtract a pseudo value function, which has 
two nodes with the same belief state $s_m$, and one is scheduled while the other 
is not. We can group the pseudo value functions, and apply the results of Case 
$1$~and Case~$2$ above. In particular, for the pseudo value functions in the 
first line of (\ref{seq:upb5_3}), the belief vectors are equal but for $s_j$ and 
$s_m$, moreover $j\in\mathcal{S}_\ordpiv^{K}$ and $m \in 
\tilde{\mathcal{S}}_\ordpiv^{K}$, and $s_j \geq s_m$, so we can apply the 
results of Case $1$.  Similarly, for the two pseudo value functions in the 
second line of (\ref{seq:upb5_3}) we can use the results of Case~$2$.  
\begin{floateq}
    \begin{subequations}\label{seq:upb5_all}
    \begin{IEEEeqnarray}[\normalsize]{rcl}
  &W_n&(s_{\ordpiv(1)},\dots,s_j,\ldots,s_{\ordpiv(K)},s_m,\ldots,s_{\ordpiv(N)}) 
- W_n(\tilde{s}_{\ordpiv(1)},\dots,\tilde{s}_m,\tilde{s}_{\ordpiv(K+1)},\dots, 
\tilde{s}_j,\ldots,\tilde{s}_{\ordpiv(N)}) \label{seq:upb5_1}\nonumber \\
  &&=W_n(s_{\ordpiv(1)},\dots,\!s_j,\ldots,s_{\ordpiv(K)},\dots,s_m,\ldots,s_{
\ordpiv(N)})\! -\!\! W_n(s_{\ordpiv(1)},\dots,\!s_m,\ldots,s_{\ordpiv(K)},\dots, 
\tilde{s}_j,\ldots,s_{\ordpiv(N)}) \label{seq:upb5_2}\\  
&&=W_n(s_{\ordpiv(1)},\dots,\!s_j,\ldots,s_{\ordpiv(K)},\dots,\!s_m,\ldots,s_{
\ordpiv(N)})\!-\!\! 
W_n(s_{\ordpiv(1)},\dots,\!s_m,\ldots,s_{\ordpiv(K)},\dots,\!s_m,\ldots,s_{
\ordpiv(N)}) \nonumber \\
 && +W_n(s_{\ordpiv(1)},\dots,\!s_m,\ldots,s_{\ordpiv(K)},\dots,\!s_m,\ldots,s_{
\ordpiv(N)})\!-\!\! W_n(s_{\ordpiv(1)},\dots,\!s_m,\ldots,s_{\ordpiv(K)},\dots, 
\!\tilde{s}_j,\ldots,s_{\ordpiv(N)})\label{seq:upb5_3}\\
  && \leq \Delta_{u} (s_j-s_m) u(n) + \Delta_{u} (s_m-\tilde{s}_j)u(n)  
\label{seq:upb5_4}\\
   &&= \Delta_{u} (s_j-\tilde{s}_j) u(n). \label{seq:upb5_5} 
   \end{IEEEeqnarray}
    \end{subequations}
    \end{floateq}
\end{itemize}

\begin{floateq}
\begin{subequations} \label{eq:lem:swap}
\begin{IEEEeqnarray}[\normalsize]{rCL}
  &W_n&(\mathbf{s})- W_n(\mathbf{s}_\Pi) \nonumber \\
  &=& R(\mathbf{s}^K) -R(\mathbf{s}^K_\Pi) + 
\beta\displaystyle\!\!\!\!\!\sum_{\mathcal{E}\subseteq  
\mathcal{S}\backslash\{j\}}  \!\!\!\!q(|\mathcal{E}|,K\!-\!1)  \Big( p 
W_{n +1} \big([\mathrm{T}(\mathbf{s}_{\overline{\mathcal{E} \cup 
j}}),\mathbf{0}(|\mathcal{E}| \!+\!1) ] \big ) +(1-p) W_{n +1} 
\big([\mathrm{T}(\mathbf{s}_{\overline{\mathcal{E}}}),\mathbf{0}(|\mathcal{E}| ) 
]\big)  \nonumber \\
  && \qquad\qquad\qquad\qquad\qquad\qquad\qquad- p W_{n +1} \big([\mathrm{T}( 
\mathbf{s}_{\overline{\mathcal{E}\cup i}}),\mathbf{0}(|\mathcal{E}| 
\!+\!1)]\big)-(1-p) W_{n +1} 
\big([\mathrm{T}(\mathbf{s}_{\overline{\mathcal{E}}}),\mathbf{0}(|\mathcal{E}|)]
\big) \Big) \label{eq:lem:swap:1}\\
  &=& R(\mathbf{s}^K) -R(\mathbf{s}^K_\Pi) - 
p\beta\displaystyle\!\!\!\!\!\sum_{\mathcal{E}\subseteq 
\mathcal{S}\backslash\{j\}}   \!\! \!\!\!\! q(|\mathcal{E}|,K\!-\!1)  
\Big(W_{n +1} \big([\mathrm{T}(\mathbf{s}_{\overline{\mathcal{E}\cup 
i}}),\mathbf{0}(|\mathcal{E}| \!+\!1)]\big)\!-\!W_{n +1} 
\big([\mathrm{T}(\mathbf{s}_{\overline{\mathcal{E}\cup 
j}}),\mathbf{0}(|\mathcal{E}| \!+\!1) ] \big )  \Big) \label{eq:lem:swap:1b}\\
  &\geq&  \Delta_{l}(s_j-s_i) - 
p\beta\displaystyle\!\!\!\!\!\sum_{\mathcal{E}\subseteq 
\mathcal{S}\backslash\{j\}}  \!\! \!\! q(|\mathcal{E}|,K\!-\!1)  
\Big(W_{n +1} \big([\mathrm{T}(\mathbf{s}_{\overline{\mathcal{E}\cup 
i}}),\mathbf{0}(|\mathcal{E}| \!+\!1)]\big) -W_{n +1} 
\big([\mathrm{T}(\mathbf{s}_{\overline{\mathcal{E}\cup 
j}}),\mathbf{0}(|\mathcal{E}| \!+\!1) ] \big )  \Big) \label{eq:lem:swap:2} \\
  &\geq&  \Delta_{l}(s_j-s_i) - 
p\beta\displaystyle\!\!\!\!\!\sum_{\mathcal{E}\subseteq 
\mathcal{S}\backslash\{j\}}  \!\! \!\!  \Big( q(|\mathcal{E}|,K\!-\!1)  
\Delta_{u} (\tau(s_j)-\tau(s_i)) u(n+1)\Big ) \label{eq:lem:swap:3} \\
  &\geq&  \Delta_{l}(s_j-s_i) -  p\beta \Delta_{u} (s_j-s_i)u(n+1) 
\label{eq:lem:swap:4}  \\
  &\geq&  \Delta_{l}(s_j-s_i) -  p\beta \Delta_{u} (s_j-s_i)u(0) 
\label{eq:lem:swap:4b}  \\
  &=& (s_j-s_i)\left( \Delta_{l} -   
p\beta\frac{1-\beta(1-p)^{T+1}}{1-\beta(1-p)}\Delta_{u}\right) \geq 0 
\label{eq:lem:swap:5}
\end{IEEEeqnarray}
\end{subequations}
\end{floateq}
 
\section{}\label{app:diff}
We note that set $\mathcal{S}=\{1,\ldots,K\}$ is the set of $K$ nodes  scheduled 
from  $\mathbf{s}$, and that the set $\mathcal{S}_\Pi^{K}$ is the set of nodes 
scheduled from  $\mathbf{s}_\Pi$, that is, the first $K$ nodes as ordered 
according to permutation $\Pi$. We only need to study the cases in which 
$\mathcal{S}$ and $\mathcal{S}_\Pi^{K}$ are different, since the claim holds for 
the others due to the symmetric property of the pseudo value function, 
(\ref{eq:Wsymmetric}). We study the case $j \in \mathcal{S}$, $i \in 
\mathcal{S}_\Pi^{K}$, $i \notin \mathcal{S}$, and $j \notin 
\mathcal{S}_\Pi^{K}$ in \eqref{eq:lem:swap}. The summation in (\ref{eq:lem:swap:1}) is over all operative/inoperative 
combinations of the nodes in $\mathcal{S}\backslash \{j\}$. We denote the belief 
state of all nodes but those in $\mathcal{E}$ and $s_j$ by 
$\mathbf{s}_{\overline{\mathcal{E} \cup j}}$. The belief state of node $i$ in TS 
$n+1$, $\tau(s_i)$, is in vector $\mathrm{T}(\mathbf{s}_{\overline{\mathcal{E} 
\cup j}})$. Similarly, the belief state of node $j$ in TS $n+1$, $\tau(s_j)$, is 
in vector $\mathrm{T}(\mathbf{s}_{\overline{\mathcal{E} \cup i}})$. The second 
pseudo value functions in the first and second lines in (\ref{eq:lem:swap:1}) 
cancel out, and (\ref{eq:lem:swap:1b}) is obtained.  We have applied the 
decomposability and boundedness of $R(\cdot)$ to obtain (\ref{eq:lem:swap:2}). 
Belief vectors $\mathrm{T}(\mathbf{s}_{\overline{\mathcal{E}\cup j}})$ and 
$\mathrm{T}(\mathbf{s}_{\overline{\mathcal{E}\cup i}})$ in (\ref{eq:lem:swap:2}) 
are ordered and only differ in one element, $\tau(s_i)$ and $\tau(s_j)$, 
respectively, where  $\tau(s_i)\leq \tau(s_j)$, and hence, we use 
Lemma~\ref{lem:upb} to get (\ref{eq:lem:swap:3});  (\ref{eq:lem:swap:4}) follows 
since $\tau(\cdot)$ is a monotonically increasing contracting map, 
(\ref{eq:lem:swap:4b}) since $u(n)$ is decreasing in $n$; finally 
(\ref{eq:lem:swap:5}) follows since  $u(0)$ is the sum of a geometric series.

\section{}\label{app:case2}
We again use backward induction. Lemma~\ref{lem:case2} holds trivially for 
$n=T$. Note that in (\ref{eq:lem:case2:1}) the set of nodes scheduled in the 
pseudo value functions $W_n(\mathbf{s}_\barpi)$ and $W_n(\mathbf{s})$ are 
$\{1,\ldots,K-1,N\}$ and $\{1,\ldots,K\}$, respectively. That is, node $K$ is 
scheduled in $W_n(\mathbf{s})$, but not in $W_n(\mathbf{s}_\barpi)$; and node 
$N$ is scheduled in $W_n(\mathbf{s}_\barpi)$, but not in  $W_n(\mathbf{s})$.  To 
prove that (\ref{eq:lem:case2:1}) holds at TS $n$ we use a sample path argument 
similarly to \cite{RMAB:Ahmad2009a}, and assume that the realizations of the EH 
processes of nodes $K$ and  $N$ are either $0$ or $1$. There are four different 
cases, but here we only consider one, since the others follow similarly.  

We consider the case in which the EH processes have realizations $E^s_K(n)=1$ 
and $E^s_N(n)=0$.  We denote by $\mathcal{K}=\{1,\ldots,K-1\}$ the set of   
nodes scheduled in both sides of (\ref{eq:lem:case2:1}). If $\mathcal{E}$ is the 
set of active nodes,  we denote the set of nodes in $\mathcal{K}$ that remain 
idle by $\mathcal{K}^i=\mathcal{K}\backslash \mathcal{E}$. We denote  the nodes 
that are not scheduled in either side of (\ref{eq:lem:case2:1}) by 
$\mathcal{K}^s=\overline{\mathcal{K}\cup \{K, N\}}$.  We denote the set 
$\{0,1\}^{|\mathcal{E}|}$ by $\mathcal{B}^{|\mathcal{E}|}$. From the left hand 
side of (\ref{eq:lem:case2:1}) we obtain \eqref{seq_case2_part1}, 
where in  \eqref{seq_case2_part1_si} we have applied the induction hypothesis of 
(\ref{eq:lem:case2:1}), the symmetry of the pseudo value function, the 
inequality $p_{11} \geq p_{00}$, and the definition of  $R(\cdot)$. This 
concludes the proof of (\ref{eq:lem:case2:1}). 
\begin{floateq} 
\begin{subequations}\label{seq_case2_part1}
\begin{IEEEeqnarray}[\normalsize]{rCL} 
1&+&W_n(s_N,s_1,\ldots, s_{N-1}) \nonumber\\
&=& 1+R(\mathbf{s}_\barpi^K)+\beta\!\!\displaystyle\sum_{\mathcal{E} \subseteq 
\mathcal{K}}\!\!\!\!\!\!\!\displaystyle\sum_{~~~~l_\mathcal{E}\in 
\mathcal{B}^{|\mathcal{E}|}}\!\!\!\!\!\!h(l_\mathcal{E},K\!-\!1) \Big [p W_{n+1} 
\left (\mathbf{P}_{11}\left (\Sigma l_\mathcal{E} \right ),  
\pmb{\tau}(\mathbf{s}_{\mathcal{K}^i} ), s_K= p_{11},  
\pmb{\tau}(\mathbf{s}_{\mathcal{K}^s} 
),s_N=p_{01},\mathbf{P}_{01}\left(\overline{\Sigma} l_\mathcal{E}\right ) \right 
) \nonumber\\
&&\qquad\qquad\qquad\qquad+(1-p)W_{n+1} \left (\mathbf{P}_{11}\left (\Sigma 
l_\mathcal{E} \right ), s_N=p_{01}, \pmb{\tau}(\mathbf{s}_{\mathcal{K}^i} ), 
s_K=p_{11},  \pmb{\tau}(\mathbf{s}_{\mathcal{K}^s}),   
\mathbf{P}_{01}\left(\overline{\Sigma} l_\mathcal{E}\right ) \right ) \Big] \\
&\geq&  p+R(\mathbf{s}_\barpi^K)+\beta\!\!\displaystyle\sum_{\mathcal{E} 
\subseteq \mathcal{K}}\!\!\!\!\!\!\!\displaystyle\sum_{~~~~l_\mathcal{E}\in 
\mathcal{B}^{|\mathcal{E}|}}\!\!\!\!\!\!h(l_\mathcal{E},K\!-\!1) \Big [p W_{n+1} 
\left (\mathbf{P}_{11}\left (\Sigma l_\mathcal{E} \right ),  
\pmb{\tau}(\mathbf{s}_{\mathcal{K}^i} ),s_K= p_{11},  
\pmb{\tau}(\mathbf{s}_{\mathcal{K}^s} 
),s_N=p_{01},\mathbf{P}_{01}\left(\overline{\Sigma} l_\mathcal{E}\right ) \right 
) \nonumber\\
&&\qquad\qquad\qquad\qquad+(1-p)\Big(1+W_{n+1} \left (\mathbf{P}_{11}\left 
(\Sigma l_\mathcal{E} \right ), s_N=p_{01}, 
\pmb{\tau}(\mathbf{s}_{\mathcal{K}^i} ), s_K=p_{11},  
\pmb{\tau}(\mathbf{s}_{\mathcal{K}^s}),   \mathbf{P}_{01}\left(\overline{\Sigma} 
l_\mathcal{E}\right ) \Big )\right ) \Big] \\
&\geq& p+R(\mathbf{s}_\barpi^K)+\beta\!\!\displaystyle\sum_{\mathcal{E} 
\subseteq \mathcal{K}}\!\!\!\!\!\!\!\displaystyle\sum_{~~~~l_\mathcal{E}\in 
\mathcal{B}^{|\mathcal{E}|}}\!\!\!\!\!\!h(l_\mathcal{E},K\!-\! 1) \Big [p 
W_{n+1} \left (\mathbf{P}_{11}\left (\Sigma l_\mathcal{E} \right ),  
\pmb{\tau}(\mathbf{s}_{\mathcal{K}^i} ),s_K= p_{11},  
\pmb{\tau}(\mathbf{s}_{\mathcal{K}^s} 
),s_N=p_{01},\mathbf{P}_{01}\left(\overline{\Sigma} l_\mathcal{E}\right ) \right 
) \nonumber\\
&&\qquad\qquad\qquad\qquad+(1-p)W_{n+1} \left (\mathbf{P}_{11}\left (\Sigma 
l_\mathcal{E} \right ), \pmb{\tau}(\mathbf{s}_{\mathcal{K}^i} ), s_K=p_{11},  
\pmb{\tau}(\mathbf{s}_{\mathcal{K}^s}),   \mathbf{P}_{01}\left(\overline{\Sigma} 
l_\mathcal{E}\right ), s_N= p_{01}\right ) \Big] \label{seq_case2_part1_si}\\
&=& R(\mathbf{s}^K)+\beta\!\!\displaystyle\sum_{\mathcal{E} \subseteq 
\mathcal{K}}\!\!\!\!\!\!\!\displaystyle\sum_{~~~~l_\mathcal{E}\in 
\mathcal{B}^{|\mathcal{E}|}}\!\!\!\!\!\! h(l_\mathcal{E},K\!-\!1) \Big [p 
W_{n+1} \left (\mathbf{P}_{11}\left (\Sigma l_\mathcal{E} \right ), s_K=p_{11},  
\pmb{\tau}(\mathbf{s}_{\mathcal{K}^i} ),  \pmb{\tau}(\mathbf{s}_{\mathcal{K}^s} 
),s_N=p_{01},\mathbf{P}_{01}\left(\overline{\Sigma} l_\mathcal{E}\right ) \right 
) \nonumber\\
&&\qquad\qquad\qquad\qquad+(1-p)W_{n+1} \left (\mathbf{P}_{11}\left (\Sigma 
l_\mathcal{E} \right ),\pmb{\tau}(\mathbf{s}_{\mathcal{K}^i } ), s_K=p_{11},   
\pmb{\tau}(\mathbf{s}_{\mathcal{K}^s}), s_N=p_{01},   
\mathbf{P}_{01}\left(\overline{\Sigma} l_\mathcal{E}\right )\right ) \Big] \\
&=&W_n(\mathbf{s}) 
 \end{IEEEeqnarray}
 \end{subequations}
 \end{floateq}
\begin{floateq} 
\begin{subequations}\label{seq_case2_part2}
\begin{IEEEeqnarray}[\normalsize]{rCL}W_n(\tilde{\mathbf{s}}) &=& 
R(\tilde{\mathbf{s}}^K)+\beta\!\!\displaystyle\sum_{\mathcal{E} \subseteq 
\mathcal{K}}\!\!\!\!\!\!\!\!\displaystyle\sum_{~~~~l_\mathcal{E}\in 
\mathcal{B}^{|\mathcal{E}|}}\!\!\!\!\!\!h(l_\mathcal{E},K\!-\!1) \Big [
p W_{n+1} \left (\mathbf{P}_{11}\left (\Sigma l_\mathcal{E} \right ), 
s_j=p_{11},  \pmb{\tau}(\mathbf{s}_{\mathcal{K}^i} ),  
\pmb{\tau}(\mathbf{s}_{\mathcal{K}^s\cup i} 
),\mathbf{P}_{01}\left(\overline{\Sigma} l_\mathcal{E}\right ) \right ) 
\nonumber\\
&&\qquad\qquad \qquad \qquad\qquad \qquad+(1-p)W_{n+1} \left 
(\mathbf{P}_{11}\left (\Sigma l_\mathcal{E} \right ), 
\pmb{\tau}(\mathbf{s}_{\mathcal{K}^i \cup j} ),  
\pmb{\tau}(\mathbf{s}_{\mathcal{K}^s\cup i}),   
\mathbf{P}_{01}\left(\overline{\Sigma} l_\mathcal{E}\right ) \right ) \Big]  \\
&\geq& 
R(\tilde{\mathbf{s}}^K)-p+\beta\!\!\displaystyle\sum_{\mathcal{E} \subseteq 
\mathcal{K}}\!\!\!\!\!\!\!\!\displaystyle\sum_{~~~~l_\mathcal{E}\in 
\mathcal{B}^{|\mathcal{E}|}}\!\!\!\!\!\!h(l_\mathcal{E},K\!-\!1) \Big [p \Big ( 
1+ W_{n+1} \left (\mathbf{P}_{11}\left (\Sigma l_\mathcal{E} \right ), 
s_i=p_{01},  \pmb{\tau}(\mathbf{s}_{\mathcal{K}^i} ),  
\pmb{\tau}(\mathbf{s}_{\mathcal{K}^s\cup j} 
),\mathbf{P}_{01}\left(\overline{\Sigma} l_\mathcal{E}\right ) \right ) 
\Big)\nonumber\\
&&\qquad\qquad \qquad \qquad \qquad\qquad+(1-p)W_{n+1} \left 
(\mathbf{P}_{11}\left (\Sigma l_\mathcal{E} \right ), 
\pmb{\tau}(\mathbf{s}_{\mathcal{K}^i \cup i} ),  
\pmb{\tau}(\mathbf{s}_{\mathcal{K}^s\cup j}),   
\mathbf{P}_{01}\left(\overline{\Sigma} l_\mathcal{E}\right ) \right ) \Big] \label{seq_case2_part2_i1}\\
&\geq&  
R(\tilde{\mathbf{s}}^K)-p+\beta\!\!\displaystyle\sum_{\mathcal{E} \subseteq 
\mathcal{K}}\!\!\!\!\!\!\!\!\displaystyle\sum_{~~~~l_\mathcal{E}\in 
\mathcal{B}^{|\mathcal{E}|}}\!\!\!\!\!\!h(l_\mathcal{E},K\!-\!1) \Big [p W_{n+1} 
\left (\mathbf{P}_{11}\left (\Sigma l_\mathcal{E} \right ),  
\pmb{\tau}(\mathbf{s}_{\mathcal{K}^i} ),  
\pmb{\tau}(\mathbf{s}_{\mathcal{K}^s\cup j} 
),\mathbf{P}_{01}\left(\overline{\Sigma} l_\mathcal{E}\right ), s_i=p_{01} 
\right ) \nonumber\\
&&\qquad\qquad\qquad \qquad \qquad \qquad+(1-p)W_{n+1} \left 
(\mathbf{P}_{11}\left (\Sigma l_\mathcal{E} \right ), 
\pmb{\tau}(\mathbf{s}_{\mathcal{K}^i \cup i} ),  
\pmb{\tau}(\mathbf{s}_{\mathcal{K}^s\cup j}),   
\mathbf{P}_{01}\left(\overline{\Sigma} l_\mathcal{E} \right ) \right ) \Big] \label{seq_case2_part2_i2}\\
&=&W_n(\tilde{\mathbf{s}}_\Pi)
 \end{IEEEeqnarray}
 \end{subequations}
 \end{floateq}

Now we prove the second part of Lemma~\ref{lem:case2}, (\ref{eq:lem:case2:2}). 
There are three cases:
\begin{itemize}
 \item Case $1$: $j,i \leq K$, i.e.,  nodes $j$ and $i$ are scheduled on both 
sides of (\ref{eq:lem:case2:2}). The inequality holds since the pseudo value 
function is symmetric. 
 \item Case $2$: $j\leq K$ and $i> K$ in \eqref{seq_case2_part2}, i.e.,  nodes $i$ and $j$ are scheduled on 
the left and right hand sides of (\ref{eq:lem:case2:2}), respectively. To prove 
the inequality we use the linearity of the pseudo value function 
(\ref{eq:pseudo_linprop}). Since $s_j\geq s_i$, using (\ref{eq:pseudo_linprop}), 
we only need to prove that $W_n(s_1,\ldots, 1, \ldots, 0, \ldots, 
s_N)-W_n(s_1,\ldots, 0, \ldots, 1, \ldots, s_N)\geq 0$. 
  We denote  the scheduled nodes in both sides of (\ref{eq:lem:case2:2}) by 
$\mathcal{K}=\{1,\ldots,K\}\backslash \{j\}$,  the set of nodes in $\mathcal{K}$ 
that remain idle by $\mathcal{K}^i=\mathcal{K}\backslash \mathcal{E}$, and  the 
nodes that are not scheduled in either side of (\ref{eq:lem:case2:2}) by 
$\mathcal{K}^s=\overline{\mathcal{K}\cup\{j,i\}}$.  We denote the belief vector 
$(s_1,\ldots, s_j=1, \ldots, s_i=0, \ldots, s_{N})$ by $\tilde{\mathbf{s}}$, its $i,j$-swap by $\tilde{\mathbf{s}}_\Pi$, and define $\tilde{\mathbf{s}}^K \triangleq (\tilde{s}_1,\ldots,\tilde{s}_K)$. In \eqref{seq_case2_part2} have used the 
induction hypothesis of (\ref{eq:lem:case2:2}) and (\ref{eq:lem:case2:1}) in \eqref{seq_case2_part2_i1}
 and \eqref{seq_case2_part2_i2}, respectively, and the fact that $\beta \leq 1$. 
\item Case $3$: nodes $s_j$ and $s_i$ are not scheduled. Inequality  holds in 
this case, by applying the definition of (\ref{eq:pseudo_val_case2}) and the 
induction hypothesis of (\ref{eq:lem:case2:2}).

\end{itemize}
\end{appendices}
\bibliographystyle{IEEEtran} 
 \bibliography{IEEEabrv,Totabiblio}

\begin{IEEEbiography}[{\includegraphics[width=1in,height=1.25in,clip,
keepaspectratio]{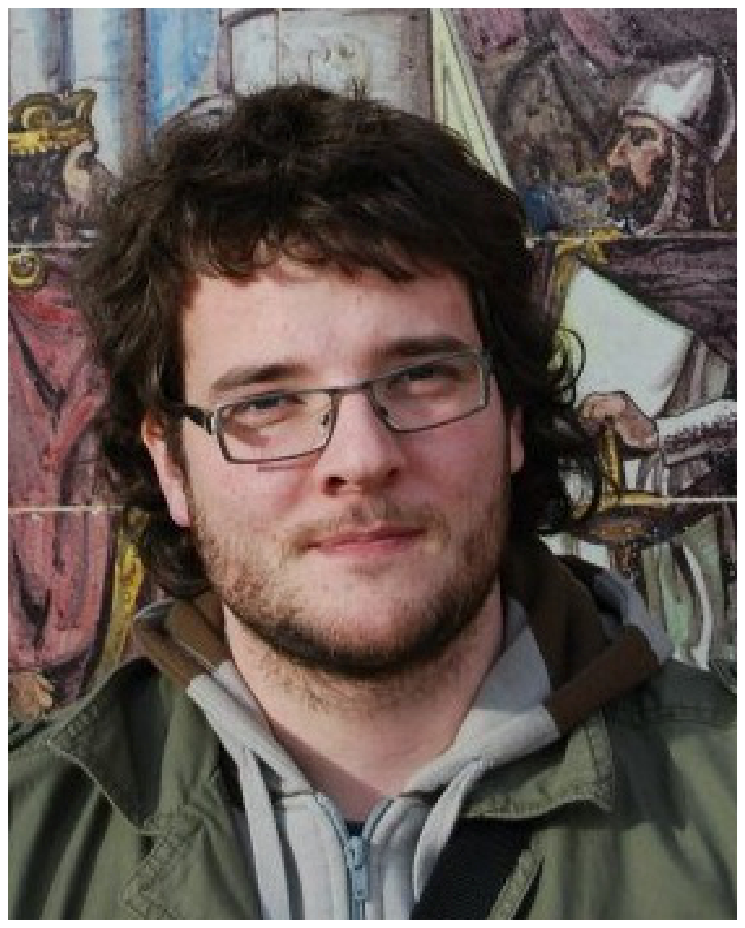}}]{Pol Blasco}
received the B.Eng. degree from Technische Universit\"{a}t Darmstadt, Germany, and 
BarcelonaTech (formally UPC), Spain, in 2008 and 2009, respectively, the M. 
S. degree from BarcelonaTech, in 2011, and the Ph.D. degree in 
electrical engineering from Imperial College London, in 2014. He was research 
assistant at CTTC in Barcelona, Spain,  from November 2009 until August 2013. He 
was a visiting scholar in the Centre for Wireless Communication in University of 
Oulu, Finland, during the last semester of 2011, and to Imperial College during the first half of 2013. In 2008 he pursued the B.Eng. 
thesis in European Space Operation Center in the OPS-GSS section, Darmstadt, 
Germany. He also has carried on research in neuroscience in IDIBAPS, Barcelona, 
Spain, and in the Technische Universit\"{a}t Darmstadt in collaboration with the 
 Max-Planck-Institut, Frankfurt, Germany, in 2009 and 2008, respectively. His 
current research interest cover communication of energy harvesting devices, 
cognitive radio, machine learning, control theory, decision making, and 
neuroscience.
\end{IEEEbiography}

\begin{IEEEbiography}[{\includegraphics[width=1in,height=1.25in,clip,
keepaspectratio]{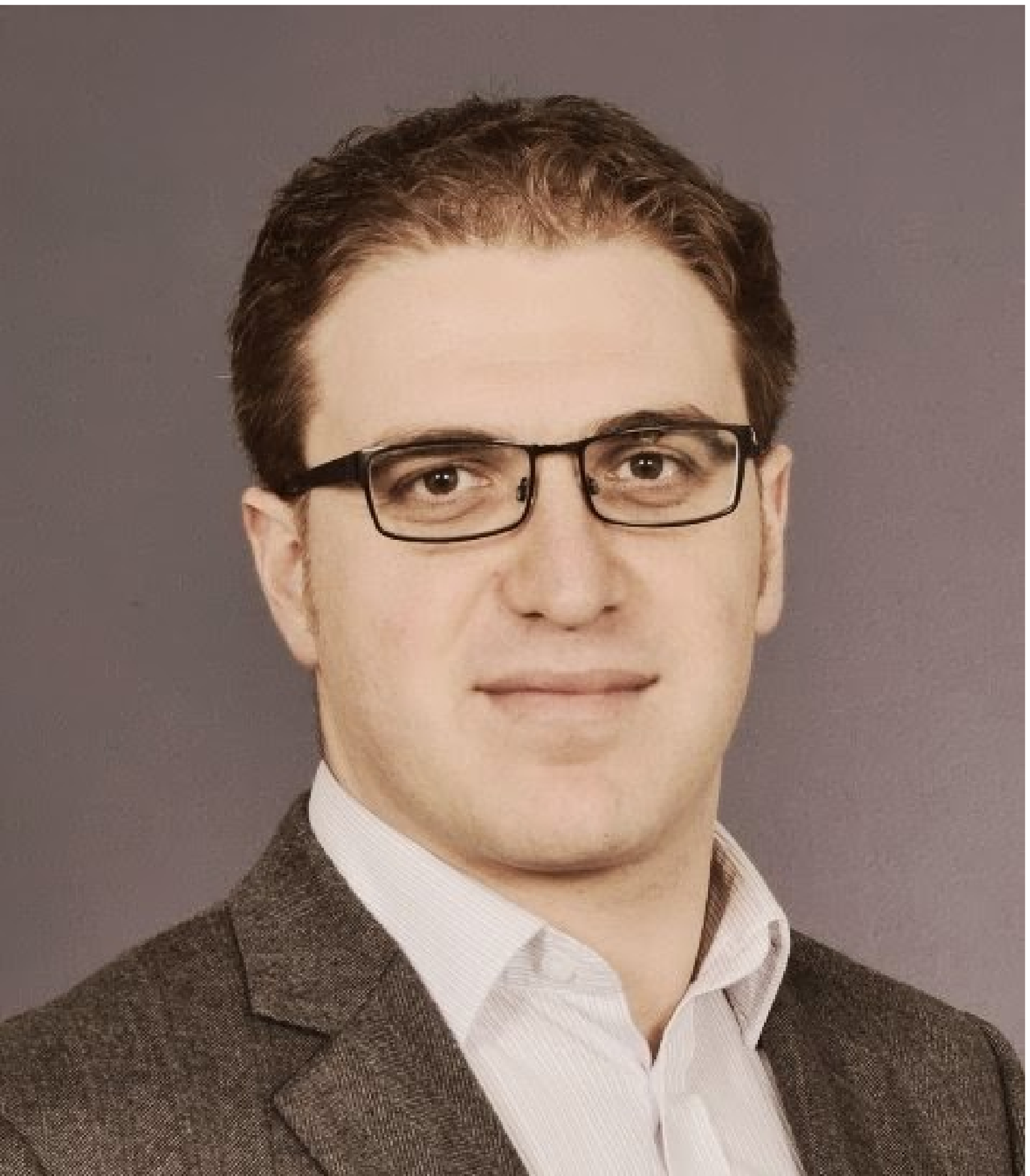}}]{Deniz G\"{u}nd\"{u}z} (S'03-M'08-SM'13) received the M.S. 
and Ph.D. degrees in electrical engineering from NYU Polytechnic School of 
Engineering in 2004 and 2007, respectively. After his PhD he served as a 
postdoctoral research associate at Princeton University, and as a consulting assistant professor at Stanford University. Since September 2012 he is a 
Lecturer in the Electrical and Electronic Engineering Department of Imperial 
College London. Previously he was a research associate at CTTC in Barcelona, 
Spain. He also held a visiting researcher position at Princeton University from 
November 2009 until November 2011.

Dr. Gunduz is an Associate Editor of the IEEE TRANSACTIONS ON COMMUNICATIONS. He is serving as a co-chair of the IEEE Information Theory Society Student Committee, and a co-director of the Imperial College Probability Center. He has served as the co-chair of the Network Theory Symposium at the 2013 and 2014 IEEE Global Conference on Signal and Information Processing (GlobalSIP), and he was a co-chair of the 2012 IEEE European School of 
Information Theory (ESIT). His research interests lie in the areas of 
communication theory and information theory with special emphasis on joint 
source-channel coding, multi-user networks, energy efficient communications and 
privacy.
\end{IEEEbiography}

\end{document}